\documentclass[11pt]{article}
\usepackage[margin=1in]{geometry}

\usepackage[utf8]{inputenc} 
\usepackage{booktabs}       
\usepackage{amsfonts}       
\usepackage{nicefrac}       
\usepackage{microtype}      
\usepackage{xcolor}         
\usepackage{parskip}        
\usepackage{amsmath, amsthm, mathtools, dsfont}
\usepackage{txfonts} 
\usepackage{color,graphicx}
\usepackage{url,hyperref}
\usepackage{enumerate}
\usepackage{algorithm}
\usepackage[noend]{algpseudocode}
\usepackage{fancybox}
\usepackage{xspace}
\usepackage{cuted,tcolorbox,lipsum}
\usepackage{multicol}
\usepackage{threeparttable}
\usepackage{xcolor}
\usepackage[T1]{fontenc}
\usepackage{tikz}
\usepackage{tcolorbox}
\usepackage{hyperref}
\usepackage[all]{hypcap}
\usetikzlibrary{calc}
\tcbuselibrary{skins}
\usepackage[colorinlistoftodos]{todonotes}
\usepackage{array,multirow}

\usepackage{amsthm}
\usepackage{framed}
\usepackage{mdframed}

\newcommand*\samethanks[1][\value{footnote}]{\footnotemark[#1]}

\theoremstyle{definition}

\colorlet{shadecolor}{gray!20}

\newtheorem{theorem}{Theorem}
\newtheorem{corollary}[theorem]{Corollary}
\newtheorem{lemma}[theorem]{Lemma}

\newtheorem{definition}[theorem]{Definition}

\newcommand{\pcstgw}{\textsc{PCSTGW}}
\newcommand{\steinertree}{\textsc{SteinerTree}}
\newcommand{\rpcst}{\textsc{IPCST}}

\newcommand{\Root}{\textit{root}}
\newcommand{\RSET}{\{\Root\}}

\newcommand{\GW}{\text{GW}}
\newcommand{\ST}{\text{ST}}
\newcommand{\IT}{\text{IT}}

\newcommand{\apx}{1.7994}

\newcommand{\alf}{\alpha}
\newcommand{\bta}{\beta}

\newcommand{\p}{p}
\newcommand{\pv}{\pi_v}
\newcommand{\pb}{\pi_\bta}

\newcommand{\I}{\ensuremath{I}}
\newcommand{\Ib}{\ensuremath{\I_\bta}}
\newcommand{\R}{\ensuremath{R}}

\newcommand{\G}{\ensuremath{G}}
\newcommand{\V}{\ensuremath{V}}
\newcommand{\E}{\ensuremath{E}}
\newcommand{\cc}{\ensuremath{c}}

\newcommand{\A}{\ensuremath{\mathcal{A}}}
\newcommand{\B}{\ensuremath{\mathcal{B}}}
\newcommand{\C}{\ensuremath{\mathcal{C}}}
\newcommand{\D}{\ensuremath{\mathcal{D}}}
\newcommand{\Dp}{\ensuremath{\mathcal{D'}}}
\newcommand{\Dz}{\ensuremath{\mathcal{D''}}}
\newcommand{\Bp}{\ensuremath{\mathcal{B'}}}
\newcommand{\Bz}{\ensuremath{\mathcal{B''}}}

\newcommand{\rr}{r}
\newcommand{\va}{\rr_{\A}}
\newcommand{\vb}{\rr_{\B}}
\newcommand{\vba}{\ensuremath{b_1}}
\newcommand{\vbb}{\ensuremath{b_2}}
\newcommand{\vc}{\rr_{\C}}
\newcommand{\vd}{\rr_{\D}}
\newcommand{\vdp}{\rr_{\Dp}}
\newcommand{\vdz}{\rr_{\Dz}}
\newcommand{\vbp}{\rr_{\Bp}}
\newcommand{\vbz}{\rr_{\Bz}}

\newcommand{\CC}{\ensuremath{C}}
\newcommand{\AS}{\ensuremath{ActS}}
\newcommand{\SET}{\ensuremath{S}}

\newcommand{\dead}{\DGW}
\newcommand{\ah}{\ensuremath{H_A}}
\newcommand{\ih}{\ensuremath{H_I}}
\newcommand{\comp}{\ensuremath{comp}}
\newcommand{\DS}{\ensuremath{DS}}
\newcommand{\U}{\ensuremath{U}}

\newcommand{\OPT}{\text{OPT}}
\newcommand{\OPTp}{\ensuremath{\OPT'}}

\newcommand{\T}{\ensuremath{T}}
\newcommand{\F}{\ensuremath{F}}

\newcommand{\cost}{\textit{cost}}

\newcommand{\y}{\ensuremath{y}}
\newcommand{\ys}{\ensuremath{\y_S}}
\newcommand{\ysv}{\ensuremath{\y_{Sv}}}
\newcommand{\yv}{\ensuremath{\y_v}}

\newcommand{\TGW}{\ensuremath{\T_{\GW}}}
\newcommand{\TST}{\ensuremath{\T_{\ST}}}
\newcommand{\TIT}{\ensuremath{\T_{\IT}}}
\newcommand{\TOPT}{\ensuremath{\T_{\OPT}}}

\newcommand{\NVGW}{\ensuremath{\overline{\V(\TGW)}}}
\newcommand{\NVST}{\ensuremath{\overline{\V(\TST)}}}
\newcommand{\NVIT}{\ensuremath{\overline{\V(\TIT)}}}
\newcommand{\NVOPT}{\ensuremath{\overline{\V(\TOPT)}}}

\newcommand{\cGW}{\ensuremath{\cost_{\GW}}}
\newcommand{\cST}{\ensuremath{\cost_{\ST}}}
\newcommand{\cIT}{\ensuremath{\cost_{\IT}}}
\newcommand{\cOPT}{\ensuremath{\cost_{\OPT}}}
\newcommand{\cOPTR}{\ensuremath{\cost_{\OPT_{\R}}}}

\newcommand{\DGW}{\ensuremath{K}}
\newcommand{\NDGW}{\ensuremath{L}}

\newcommand{\wx}{\ensuremath{w_{\GW}}}
\newcommand{\wy}{\ensuremath{w_{\ST}}}
\newcommand{\wz}{\ensuremath{w_{\IT}}}
\newcommand{\cWAG}{\ensuremath{\cost_{\text{WAG}}}}

\newcommand{\dopt}{d_{\OPT}}

\title{Prize-Collecting Steiner Tree: A 1.79 Approximation}
\date{}

\author{
Ali Ahmadi\thanks{University of Maryland.}\\
\texttt{ahmadia@umd.edu}
\and
Iman Gholami\samethanks\\
\texttt{igholami@umd.edu}
\and
MohammadTaghi Hajiaghayi\samethanks\\
\texttt{hajiagha@umd.edu}
\and 
Peyman Jabbarzade\samethanks\\
\texttt{peymanj@umd.edu}
\and
Mohammad Mahdavi\samethanks\\
\texttt{mahdavi@umd.edu}
}

\begin{document}
\maketitle


\begin{abstract}
Prize-Collecting Steiner Tree (PCST) is a generalization of the Steiner Tree problem, a fundamental problem in computer science. 
In the classic Steiner Tree problem, we aim to connect a set of vertices known as terminals using the minimum-weight tree in a given weighted graph. 
In this generalized version, each vertex has a penalty, and there is flexibility to decide whether to connect each vertex or pay its associated penalty, making the problem more realistic and practical.

Both the Steiner Tree problem and its Prize-Collecting version had long-standing $2$-approximation algorithms, matching the integrality gap of the natural LP formulations for both. 
This barrier for both problems has been surpassed, with algorithms achieving approximation factors below $2$. 
While research on the Steiner Tree problem has led to a series of reductions in the approximation ratio below $2$, culminating in a $\ln(4)+\epsilon$ approximation by Byrka, Grandoni, Rothvo{\ss}, and Sanità \cite{DBLP:conf/stoc/ByrkaGRS10}, the Prize-Collecting version has not seen improvements in the past 15 years since the work of Archer, Bateni, Hajiaghayi, and Karloff \cite{DBLP:journals/siamcomp/ArcherBHK11,DBLP:conf/focs/ArcherBHK09} (FOCS'09), which reduced the approximation factor for this problem from $2$ to $1.9672$.
Interestingly, even the Prize-Collecting TSP approximation, which was first improved below $2$ in the same paper, has seen several advancements since then (see, e.g., Blauth and N{\"{a}}gele~\cite{DBLP:conf/stoc/BlauthN23} in STOC'23).

In this paper, we reduce the approximation factor for the PCST problem substantially to 1.7994 via a novel iterative approach.
\end{abstract}


\section{Introduction}



The Steiner Tree problem is a well-known problem in the field of combinatorial optimization. It involves connecting a specific set of vertices (referred to as terminals) in a weighted graph while aiming to minimize the total cost of the edges used. The problem also allows for the inclusion of additional vertices, known as Steiner points, which can help reduce the overall cost. This problem has a long history and was formally defined mathematically by Hakimi in 1971 \cite{DBLP:journals/networks/Hakimi71}. It is recognized as one of the classic NP-hard problems \cite{DBLP:conf/coco/Karp72}. The Steiner Tree problem finds applications in various domains, including network design \cite{DBLP:books/daglib/0069809} and phylogenetics \cite{DBLP:journals/classification/Rohlf05}, prompting continuous research efforts to develop more efficient approximation algorithms.

Initial algorithmic strategies for the Steiner Tree problem, while heuristic in nature, set the stage for more precise approaches. Zelikovsky's 1993 introduction of a polynomial-time approximation algorithm achieved an 11/6-approximation ratio \cite{DBLP:journals/algorithmica/Zelikovsky93}, which was followed by further improvements including Karpinski and Zelikovsky's 1.65-approximation in 1995 \cite{DBLP:journals/eccc/ECCC-TR95-030}. 
The approach was refined to a 1.55-approximation by Robins and Zelikovsky in 2005 \cite{DBLP:journals/siamdm/RobinsZ05}, and by 2010, Byrka, Grandoni, Rothvoß, and Sanità advanced this to a~\mbox{1.39-approximation} \cite{DBLP:conf/stoc/ByrkaGRS10}.
An earlier MST-based 2-approximation algorithm, introduced in the early 1980s, also played a crucial role due to its simplicity \cite{DBLP:journals/acta/KouMB81}.

The computational complexity of the Steiner Tree problem has been firmly established. Bern and Plassmann showed its MAX SNP-hardness, indicating the absence of a polynomial-time approximation scheme (PTAS) for this problem unless P equals NP \cite{DBLP:journals/ipl/BernP89}. Building on this, Chlebík and Chlebíková in 2008 established a lower bound, demonstrating that approximating the Steiner Tree problem within a factor of 96/95 of the optimal solution is NP-hard.
This finding marks a crucial step in understanding the inherent complexity of the problem \cite{DBLP:journals/tcs/ChlebikC08}.

In combinatorial optimization, prize-collecting variants are distinct for their detailed decision-making approach. These variants focus not only on building an optimal structure but also on intentionally excluding certain components, which leads to a penalty. This introduces more complexity and makes these problems more applicable to real-world scenarios. The concept of prize-collecting problems in optimization was first brought forward by Balas in the late 1980s \cite{Balas89}. This pioneering work opened a new research direction, particularly in scenarios where avoiding certain elements results in penalties. Following this, the first approximation algorithms for prize-collecting problems were introduced in the early 1990s by authors including Bienstock, Goemans, Simchi-Levi, and Williamson \cite{BienstockGSW93}. Their initial contributions have significantly shaped the research direction in this area, focusing on developing solutions that effectively balance costs against penalties.

The Prize-collecting Steiner Tree (PCST) problem is a key example in this category, as it takes into account both the costs of connectivity and penalties for excluding vertices. In this problem, we consider an undirected graph \( G = (V, E) \) where \( V \) represents vertices and \( E \) represents edges. Each edge \( e \in E \) has an associated cost \( c(e) \), and each vertex \( v \in V \) comes with a penalty \( \pi(v) \)  that needs to be paid if the vertex is not connected in the solution. The objective is to find a tree \( T = (V_T, E_T) \) within \( G \) that minimizes the sum of edge costs in \( T \) and penalties for vertices not in \( T \). This is mathematically expressed as:

\[
\text{Minimize} \quad \sum_{e \in E_T} c(e) + \sum_{v \in V \setminus V_T} \pi(v).
\]

This formulation captures the essence of the PCST problem: a trade-off between the infrastructure cost, represented by the sum of the edge costs within the chosen tree, and the penalties assigned to vertices excluded from this connecting structure.
This detailed view of the problem applies to various situations, such as network design where not every node needs to be connected, and resource allocation where some demands might not be met, resulting in a cost.

Initial strides in developing approximation algorithms for PCST were made by Bienstock, Goemans, Simchi-Levi, and Williamson with a 3-approximation achieved through linear programming relaxation \cite{BienstockGSW93}. Subsequent advancements by Goemans and Williamson, and later by Archer, Bateni, Hajiaghayi, and Karloff, refined the approximation ratio to 2 and 1.967, respectively \cite{DBLP:journals/siamcomp/GoemansW95, DBLP:conf/focs/ArcherBHK09}.
Our work contributes to the ongoing research efforts in the field by presenting a 1.7994-approximation algorithm for the PCST problem, improving upon the previous best-known ratio of 1.967 established in 2009~\cite{DBLP:conf/focs/ArcherBHK09}. This achievement marks progress in enhancing the efficiency of solutions for this long-standing open problem.

Besides PCST, the Prize-collecting Steiner Forest (PCSF) problem stands as another open area of research in combinatorial optimization. In PCSF, the objective is to efficiently connect pairs of vertices, each of which has an associated penalty for remaining unconnected. Work on this area began with the work of Agrawal, Klein, and Ravi \cite{AKRSTOC91, DBLP:journals/siamcomp/AgrawalKR95}. Following this, 3-approximation algorithms were developed using cost-sharing and iterative rounding, respectively \cite{DBLP:conf/soda/GuptaKLRS07, DBLP:conf/latin/HajiaghayiN10}. Progress continued with Hajiaghayi and Jain's 2.54-approximation algorithm \cite{DBLP:conf/soda/HajiaghayiJ06}, and more recently, the 2-approximation by Ahmadi, Gholami, Hajiaghayi, Jabbarzade, and Mahdavi \cite{DBLP:journals/corr/abs-2309-05172}.

Another related problem, the Prize-collecting version of the classic Traveling Salesman Problem (PCTSP), focuses on optimizing the length of the route taken while also accounting for penalties associated with unvisited cities.
Although the natural LP formulations for PCTSP and PCST share lots of similarities, PCTSP has experienced considerably more progress.
The first breakthrough in breaking the barrier of $2$ for PCST also introduced a $1.98$-approximation algorithm for PCTSP \cite{DBLP:conf/focs/ArcherBHK09}.
Subsequently, Goemans improved this to a $1.91$ approximation factor \cite{DBLP:journals/corr/abs-0910-0553}.
The approximation factor was further improved to $1.774$ by Blauth and Nägele \cite{DBLP:conf/stoc/BlauthN23}, and most recently, to $1.599$ by Blauth, Klein, and Nägele \cite{DBLP:journals/corr/abs-2308-06254}. 
These advances in PCSF and PCTSP underline the significance and continuous research interest in prize-collecting problems.


\subsection{Contribution Overview}

In this paper, we focus on rooted PCST where a designated vertex, denoted as $\Root$, must be included in the solution tree. The objective is to connect other vertices to $\Root$ or pay their penalty. The general PCST and its rooted variant are equivalent. Solving the general PCST involves iterating over all vertices as potential roots and solving the rooted variant for each. Conversely, we can adapt the general version to address rooted PCST by assigning an infinite penalty to the root vertex, ensuring its inclusion in the optimal solution. This two-way equivalence is crucial for our approach, allowing us to concentrate on rooted PCST and extend our findings to the general case. In the rooted version, we define an instance of the PCST problem using a graph $\G=(\V,\E,\cc)$ with edge weight function $\cc: \E \to \mathbb{R}_{\ge 0}$, root vertex $\Root$, and penalty function $\pi: \V \rightarrow \mathbb{R}_{\ge 0}$. In the penalty function, while only non-root vertices have actual penalties, we include $\Root$ in the domain of $\pi$ and assume it has penalty $\pi(\Root)=\infty$. This does not affect the actual costs of solutions, but simplifies our statements by adding consistency.

In designing our algorithm, we utilize the recursive approach introduced by~\cite{DBLP:journals/corr/abs-2309-05172}. 
The concept involves running a baseline algorithm with a higher approximation factor on PCST to get an initial solution. 
We then account for the penalties associated with any vertices identified by the baseline algorithm, paying these penalties, and subsequently removing their penalties from consideration.
Next, we apply a Steiner tree algorithm to the remaining vertices to obtain another solution. We then call our algorithm recursively with the adjusted penalties. 
At each recursive step, two algorithms are executed on the current input, each producing a tree as a solution. Our procedure aggregates all solutions generated during the recursion process and selects the one with the lowest cost as the final output.

We give a quick overview of the major components of our algorithm here.
\paragraph{Goemans and Williamson Algorithm for PCST.} 
We use a slightly modified version of the algorithm introduced by Goemans and Williamson in \cite{DBLP:journals/siamcomp/GoemansW95} as the baseline algorithm in the recursive process. We briefly present this algorithm for completeness. Throughout the paper, we refer to this algorithm as $\pcstgw$ and denote the solution found by the algorithm as $\GW$.

Let's assume that each edge of the input graph $\G$ is a curve with a length equal to its cost.
We want to build a spanning tree $\F$, which starts as a forest during our algorithm and transforms into a tree by the end of the algorithm. 
We then remove certain edges from this tree to obtain our final tree $\T$ and pay penalties for every vertex outside $\T$.

To run our algorithm, we define $\CC$ as the connected components of $\F$, and active sets $\AS$ as subsets of $\CC$. 
Initially, both $\CC$ and $\AS$ consist of single-member sets, with each vertex belonging to exactly one set. 
We assign a unique color to each vertex of the graph, with the value $\pi(v)$ representing the total duration that color $v$ can be used. 
As $\pi(\Root)=\infty$, the color of $\Root$ can be used without any limitation.

At any moment, each active set colors its adjacent edges (edges with exactly one endpoint in that set) with the color of one of its vertices that still has available color.

Every time an edge becomes fully colored, it will be added to $\F$, and subsequently, the connected components of $\F$ and active sets will be updated. 
Moreover, if all vertices in an active set run out of color, the active set becomes deactivated and will be considered a dead set, along with all the vertices inside it. 
We continue this process until all vertices are connected to the root. 
Note that this is ensured since the root has an infinite amount of color.

After the completion of this process, we remove some edges from $\F$ to obtain $\T$. 
We will select every dead set $S$ that cuts exactly one edge of $\F$ and remove all vertices in $S$ from $\F$ to obtain $\T$. 
Every live vertex, which refers to vertices not marked as dead, will be connected to the $\Root$ in $\T$, along with some dead vertices.
In fact, the tree $\T$ is \textit{the smallest} subtree of $\F$ that contains all live vertices, including $\Root$, and every vertex whose color has been used in $\T$.

\paragraph{Steiner Tree Algorithm for PCST.} 
Here we want to construct a new solution $\ST$ based on the outcome of $\pcstgw$.
During the execution of the $\pcstgw$ algorithm, certain active sets and their vertices may reach a dead state, leaving them incapable of coloring edges as their vertices have used all of their colors.
In such cases, it is reasonable to pay their penalties and subsequently remove them from consideration. 
This decision makes sense, as connecting these vertices to other vertices requires excessive costs compared to their penalties.

In the $\GW$ solution, some of these dead vertices may eventually connect to the root when other active sets link to them, and we utilize these dead vertices to connect live vertices to $\Root$. 
However, in $\ST$, we pay the penalties of all dead vertices and seek a tree that efficiently connects other vertices to $\Root$. 
The problem of finding a minimum tree that connects a set of vertices to $\Root$ is known as the Steiner Tree problem, and we employ the best-known algorithm for this, assuming it has a $\p$ approximation factor, which currently is $\ln(4) + \epsilon$~\cite{DBLP:conf/stoc/ByrkaGRS10}.

Improving the approximation factor of the Steiner Tree algorithm would consequently enhance the approximation factor of our PCST algorithm. 
It's worth noting that one might suggest paying penalties only for vertices that the $\GW$ solution pays penalties for, rather than all dead vertices. 
However, the $\GW$ solution may connect all vertices to the root and influence the Steiner Tree algorithm to establish connections for every vertex. 
This constraint restricts the algorithm's flexibility in exploring alternative tree structures.

\paragraph{Iterative algorithm.} 
Now, let's explore our iterative algorithm. 
Our aim is to create an iterative procedure that results in a $\alf$-approximation algorithm for PCST. 
We will discuss the value of $\alf$ in the future.

At the initiation of our algorithm, we divide the vertex penalties by a constant factor $\bta$ to obtain $\pb$. 
The idea of altering penalties has been used in~\cite{DBLP:journals/siamcomp/ArcherBHK11}, but they focus on increasing penalties, while we decrease them.
The specific value of $\bta$ will be determined towards the conclusion of our paper. 
This determination will be based on the value of $\p$, representing the best-known approximation factor for the Steiner Tree problem, with the goal of minimizing the approximation factor $\alf$.

Now, we execute $\pcstgw$ using the modified penalties $\pb$. 
Running $\pcstgw$ on $\pb$ provides us with a tree $\TGW$, and paying the penalty of vertices outside $\TGW$ yields one solution for the input. 
Subsequently, we pay the penalty of every vertex that becomes dead during the execution of $\pcstgw$, set their penalty to zero for the remainder of our algorithm, and connect the remaining vertices using the best-known algorithm for the Steiner tree problem, denoted as $\steinertree$.
The tree generated by $\steinertree$, denoted as $\TST$, presents another solution for the input.

Then, if no vertices with a non-zero penalty become inactive in $\pcstgw$, indicating that we haven't altered the penalties of vertices at this step, we terminate our algorithm by returning the minimum cost solution between $\TGW$ and $\TST$. Otherwise, we recursively apply this algorithm to the new penalties, and refer the tree of the best solution found by the recursive approach as $\TIT$.

Finally, we select the best solution among $\TGW$, $\TST$, and $\TIT$. It's important to note that our algorithm essentially identifies two solutions at each iteration and, in the end, selects the solution with the minimum cost among all these alternatives.

In analyzing our algorithm, we focus on its initial step, specifically the first invocation of $\pcstgw$ and $\steinertree$. 
We categorize vertices based on their status in $\pcstgw$, distinguishing between those marked as dead or live, and whether their penalties have been paid in both $\pcstgw$ and the optimal solution.
Additionally, we classify active sets based on whether they color only one edge or more than one edge of the optimal solution. Through this partitioning, we derive lower bounds for the optimal solution and upper bounds for the solutions $\TGW$ and $\TST$. 
Leveraging the recursive nature of our algorithm, we establish an upper bound for the solution $\TIT$ using induction. 
Following that, we evaluate how much these solutions deviate from $\alf \cdot \cOPT$.

Next, we show that for $\bta= 1.252$ and $\alf=1.7994$, a weighted average of the cost of the three solutions is at most $\alf\cdot\cOPT$. This shows that our algorithm when using this value of $\bta$ is a $1.7994$ approximation of the optimal solution since the minimum cost is lower than any weighted average. We note that throughout our analysis, we do not know the value of $\alf$. Instead, we obtain a system of constraints involving $\alf$, $\bta$, $\p$, and the weights in the weighted average which needs to be satisfied in order for our proof steps to be valid. Then, we find a solution to this system minimizing $\alf$ to find our approximation guarantee. In this solution, we use $\p=\ln(4)+\epsilon$, using the current best approximation factor for the Steiner tree~\cite{DBLP:conf/stoc/ByrkaGRS10}. Finally, we explain the intuition behind certain parts of our algorithm, including why we need to consider all three solutions that we obtain.

\paragraph{Outline.} In Section~\ref{sec:gw}, we explain Goemans and Williamson's $2$-approximation algorithm for PCST \cite{DBLP:journals/siamcomp/GoemansW95}, using the coloring schema effectively utilized by \cite{DBLP:journals/corr/abs-2309-05172} for PCSF.
Then, in Section~\ref{sec:iterative_alg}, we present our iterative algorithm along with its analysis.
Finally, in Section~\ref{sec:intuition}, we highlight the importance of employing both algorithms in conjunction with the iterative approach to improve the approximation factor.

\subsection{Preliminaries}
Throughout our paper, we assume without loss of generality that the given graph is connected.

Let $T$ be a subgraph, then $c(T)$ denotes the total cost of edges in $T$, i.e., $c(T) = \sum_{e \in T} c(e)$.

For a subgraph $T$, we use $\V(T)$ to represent the set of vertices in $T$, and $\overline{\V(T)}$ denotes the set of vertices outside $T$.

Given a subset of vertices $S \subseteq \V$, we define $\pi(S) = \sum_{v \in S} \pi(v)$ as the sum of penalties associated with vertices in $S$.

For a PCST solution $X$, we denote its corresponding tree as $T_X$. 
Furthermore, we use $\cost_X$ to represent the total cost of $X$, defined as $c(T_X) + \pi(\overline{\V(T_X)})$.


\section{Goemans and Williamson Algorithm}
\label{sec:gw}
Here we define a slightly modified version of the algorithm initially proposed by Goemans and Williamson in~\cite{DBLP:journals/siamcomp/GoemansW95} (hereinafter the GW algorithm) for the sake of completeness of our algorithm. 
Then we use it as a building block in our algorithm in the next section. We introduce several lemmas stating the properties of the algorithm and its output. We defer the proofs of these lemmas to the appendix.

The algorithm consists of two phases. In the first phase, we simulate a continuous process of vertices growing components around themselves and coloring the edges adjacent to these components at a constant rate. In this process, we imagine each edge $e$ with weight $\cc(e)$ as a curve of length $\cc(e)$. Each vertex $v$ has a potential coloring duration equal to its penalty $\pi(v)$. 
We assume that the root vertex $\Root$ has $\pi(\Root)=\infty$, indicating infinite coloring potential.
This process of coloring will give us a spanning tree, which we will then trim in the second phase to get a final tree.

During the algorithm, we keep a forest $\F$ of tentatively selected edges, a set $\CC$ of connected components of this forest, and a subset $\AS$ of active sets in $\CC$. For each component $
\SET$ in $\CC$, we will also store its coloring duration $\ys$.
Initially, the forest $\F$ is empty, every vertex is an active set in $\CC$, and all $\ys$ values are $0$.

At any moment in the process, all active sets color their adjacent edges using the coloring potential of their vertices at the same rate. So, the amount of color on each edge is the total duration its endpoints have been in active sets.  
We define an edge as fully colored if the combined active time of its endpoints totals at least the length of the edge while they belong to different components. 
When such an edge between two sets becomes fully colored, it is added to $\F$, and the two sets containing its endpoints are merged, with their coloring potentials summed together.
An active set becomes inactive if it runs out of coloring potential. This means that this set and its subsets have used the coloring potential of all the vertices in the set.
We call an edge getting added to $\F$ or an active set becoming inactive events in the coloring process. 
It may be possible for multiple events to happen simultaneously, and in that case, we would handle them one by one in an arbitrary order. 
The addition of one edge in the order may prevent the addition of other fully colored edges. However, this can only happen if the latter edge forms a cycle in $\F$, and therefore, the resulting components are independent of the order in which we handle the events.
As the component containing $\Root$ remains active and edges are only added between different components, $\F$ will eventually become a spanning tree of $\G$. 
This marks the completion of the coloring phase.

In the second phase, we will select a subset of $\F$ as our Steiner tree and pay the penalties for the remaining vertices. 
We refer to any active set that becomes inactive as a dead set.
Throughout the first phase, we maintain dead sets in $\DS$ to utilize them in the second phase. 
We categorize vertices into dead and live, where a dead vertex is any vertex contained in at least one dead set, and all other vertices are considered live. 
We store dead vertices in $\dead$ and return them at the end of $\pcstgw$ since they are used in our iterative algorithm in the next section.
For any dead set $\SET$, if there is exactly one edge of $\F$ cut by $\SET$ (i.e., $\left|\delta(\SET)\cap\F\right|=1$), we remove this edge and all the edges in $\F$ that have both endpoints in $\SET$. 
This effectively removes $\SET$ from the tree and disconnects its vertices from the root. 
We repeat this process until no dead set with this property can be found. Figure \ref{fig:tree-dead} illustrates how dead sets may be removed.

As each operation in the second phase disconnects only the selected dead set from the root, the final result will be a tree $\T$ that contains all the live vertices, including $\Root$. We pay the penalties for the vertices outside the tree, which are all dead vertices belonging to the dead sets we removed in the second phase. Algorithm \ref{alg:gw} provides a pseudocode that implements this process.

\begin{figure}
\begin{center}
\begin{tikzpicture}[scale=1.3]
\def\r{0.1}
\foreach \i/\l/\col in {0/$\Root$/black, 1//blue, 2//black, 3//, 4//black, 5//black, 6//black, 7/Live/black} {
    \node[label={above:\l}] (\i) at (\i, 0) {};
    \draw[\col, fill=\col] (\i) circle (\r);
};
\def\col{red}
\node (9) at (3, 1) {};
\draw (9) -- (3);
\draw[\col, fill=\col] (9) circle (\r);
\draw[\col, line width=3] (9) circle (0.6);
\node[label={[text=\col]above:Dead}] at ($(9) + (0, 0.6)$) {};
\draw[\col] (9) -- ($(9)!0.6!(3)$);

\node (10) at (1, 1) {};
\draw (10) -- (1);
\draw[\col, fill=\col] (10) circle (\r);
\draw[\col, line width=3] (10) circle (0.4);
\node[label={[text=\col]above:Dead}] at ($(10) + (0, 0.4)$) {};
\draw[\col] (10) -- ($(10)!0.4!(1)$);

\def\col{blue}
\draw[line width=3] (0) -- (1) -- (2) -- (3) -- (4) -- (5) -- (6) -- (7);
\draw[\col, line width=3] (1) circle (0.6);
\node[label={[text=\col]below:Dead}] at ($(1) - (0, 0.6)$) {};
\draw[\col, line width=3] (1) --  ($(1)!0.6!(0)$);
\draw[\col, line width=3] (1) --  ($(1)!0.6!(2)$);
\draw[\col] (1) --  ($(1)!0.6!(10)$);

\node (8) at (5, -1) {};
\draw[\col, fill=\col] (8) circle (\r);
\draw (8) -- (5);
\draw[\col, line width=3] (8) circle (1.2);
\node[label={[text=\col]below:Dead}] at ($(8) - (0, 0)$) {};
\draw[\col, line width=3] (8) -- (5);
\draw[\col, line width=3] (5) -- ($(5)!0.7!(4)$);
\draw[\col, line width=3] (5) -- ($(5)!0.7!(6)$);

\end{tikzpicture}
\end{center}
\caption{Illustration of dead sets in the final tree of GW algorithm. The dead sets colored in blue cut multiple edges of $\F$, and removing them would disconnect other vertices so they are not removed. On the other hand, the dead sets colored in red can be safely removed without affecting other vertices.}
\label{fig:tree-dead}
\end{figure}
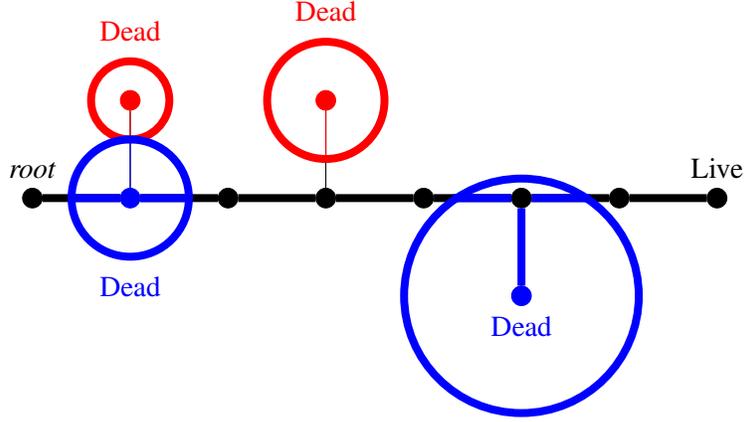

To facilitate our analysis throughout the paper, we assume that each vertex is associated with a specific color. 
During the coloring process of an active set $\SET$, we assign each moment of coloring to a vertex $v\in\SET$ with non-zero remaining coloring potential and utilize its color on the adjacent edges.
For consistency, we choose vertex $v$ based on a fixed ordering of the vertices in $\V$ where $\Root$ comes first. 
So, a set $\SET$ containing $\Root$ will always assign its coloring to $\Root$. 
We note that a set $\SET$ can not use the color of a vertex that is already dead. Based on this assignment, we define the following values:
\begin{definition}
\label{def:coldur}
For each vertex $v$, we define its total coloring duration $\yv$, and the coloring duration assigned to it by a set $\SET$ as $\ysv$: 
\begin{itemize}
    \item 
    \ysv= \text{total coloring duration using color $v$ in set $\SET$}
    \item $\yv = \sum_{\SET\subseteq \V;v\in\SET} \ysv$
\end{itemize}

Note that $\sum_{v\in\SET}\ysv = \ys$. 

\end{definition}

We bound the cost of both the chosen tree and the penalty of the dead vertices in the following lemmas. 
Proofs of these lemmas are in the appendix given their similarity to~\cite{DBLP:journals/siamcomp/GoemansW95}. 

\begin{lemma}
    \label{lemma:gwtree} 
    Let $\T$ be the tree returned by Algorithm \ref{alg:gw}. 
    We can bound the total weight of this tree by
    \begin{align*}
        \cc(\T)\leq 2\cdot\sum_{\substack{\SET \subseteq \V-\RSET;\\\SET\cap\V(\T)\neq\emptyset}}\ys=2\cdot\sum_{v \in \V(\T)-\RSET}\yv\text{.}
    \end{align*}
\end{lemma}

\begin{lemma}
    \label{lemma:gwdead} 
    For any vertex $v \in \V$, we have $\yv \leq \pi(v)$. 
    Furthermore, if $v \in \dead$ which means it is a dead vertex, we have $\yv=\pi(v)$.
\end{lemma}

\begin{lemma}
    \label{lemma:gwlive} 
    Any vertex $v \not\in \V(\T)$ is a dead vertex.
\end{lemma}

Lemmas \ref{lemma:gwtree}, \ref{lemma:gwdead}, and \ref{lemma:gwlive} immediately conclude the following lemma.

\begin{lemma}
    \label{lemma:gwbound} The total cost of the GW algorithm is bounded by
    \begin{align*}
        \cGW = \cc(\T) + \pi(\overline{\V(\T))}\leq 2\cdot\sum_{v \in \V(\T)-\RSET}\yv + \sum_{v\not\in\V(\T)} \yv\text{.}
    \end{align*}
\end{lemma}
We note that Lemma \ref{lemma:gwbound} can be used to prove that the \GW~algorithm achieves a 2-approximation by showing that the optimal solution has cost at least $\sum\limits_{v\in\V-\RSET}\yv$. We prove a stronger version of this fact in Lemma \ref{lm:opt_main_bound}.

In addition to the above lemmas on the cost of the solution and its connection to the coloring, we also prove the following lemma.
This lemma will help in our analysis in Section \ref{sec:iterative_analysis}, where we use it to introduce an upper bound for the cost of the optimal Steiner tree connecting all the live vertices in a call to the \GW~algorithm.

\begin{lemma}
\label{lem:gwlocal}
    Let $\I=(\G,\Root,\pi)$ and $\I'=(\G',\Root,\pi)$ be instances of PCST, where $\G'$ is obtained from $\G$ by adding a set of edges $\E_0$ with weight $0$ from $\Root$ to a set of vertices $\U$. Let $\yv$ be the coloring duration for vertex $v$ in a run of the \GW~algorithm on $\I$, and let $\dead$ be the set of dead vertices in this run. Let $\yv'$ and $\dead'$ be the corresponding values when running the \GW~algorithm on instance $\I'$ using the same order to assign coloring duration to vertices. We have $$\yv'\leq\yv\text{,}$$
    $$\yv'=0 \text{ if } v\in\U \text{,}$$
    and
    $$\dead'\subseteq\dead.$$
\end{lemma}

\begin{proof}
We will prove these facts by comparing the run of the GW algorithm on instances $\I$ and $\I'$. We can identify "moments" in the first phase of the algorithm in these runs by the total coloring duration using the color of $\Root$ which is always active, and look at the same moments across these two runs. Let $\CC$ and $\CC'$ refer to the set of components in the runs on $\I$ and $\I'$ respectively. We prove the invariant that at any moment in the run of the GW algorithm on $\I'$, for any component $\SET\in\CC'$ such that $\Root\not\in\SET$, $\SET$ would also be a component in $\CC$ at the same moment of the algorithm on instance $\I$. In addition, $\SET$ would be active for instance $\I'$ if and only if it is active in instance $\I$. Figure \ref{fig:inv} illustrates how the components in the runs can look.

Initially, at moment $t=0$, before any events are applied, the invariant holds as we start with each vertex being an isolated component in both cases. The invariant also holds after events at $t=0$ are processed: We can assume that the shared events are handled first in both runs, with the second run also having additional events corresponding to the edges in $\E_0$ being fully colored, which will only merge components with the component containing $\Root$.

We will now prove that if the invariant holds at moment $t$, it will also hold at the next moment $t'>t$ where an event happens in the second run. Combined with the invariant being true at time $t=0$, this will prove the invariant for the duration of the algorithm as the invariant can only break when an event occurs. Note that unless otherwise specified, when referring to a moment $t$, we consider the state of the runs after the events at moment $t$ have been applied.

Let $t$ be the current moment, where we know the invariant holds. Let $t'$ be the first moment after $t$ when an event happens in the run for instance $\I'$. We first claim that between $t$ and $t'$, there can be no events in the run for $\I$ that affect a component $\SET\in\CC'$ not containing $\Root$. Assume otherwise that such an event exists and the first event of this kind occurs at moment $t''<t'$. There are two possible cases:
\begin{itemize}
    \item The event corresponds to a fully colored edge getting added. One of the endpoints of this edge must be in set $\SET$. Let $\SET'$ be the set in $\CC'$ containing the other endpoint. If $\Root\not\in$ $\SET'$, then at each moment until $t''$, the component containing each endpoint has been the same between $\I$ and $\I'$. In addition, these components have been active at the same moments. So, the amount of coloring on the edge is the same in both runs, and this edge should become fully colored in the run on $\I'$ at time $t'$ too. This is in contradiction with the fact that the first event after moment $t$ for $\I'$ is at time $t'>t''$. 
    
    We arrive at the same contradiction if $\SET'$ includes $\Root$. In this case, the coloring on the edge would have been the same in both runs until the other endpoint joins a component including $\Root$. Afterward, the coloring from the endpoint in $\SET$ would be the same between the two runs, and the other end is always in an active set. So, the coloring on this edge in $\I'$ at moment $t''$ is at least as much as in $\I$ and so it must be fully colored by $t''$. This also can't happen before $t$ since $\SET$ is a component in $\CC'$, so we again get an event between $t$ and $t'$ which is a contradiction.
    \item The event corresponds to the set $\SET$ becoming inactive. Since $\SET$ and its subsets have been active sets at the same moments in both runs, if $\SET$ becomes inactive in $\I$ at time $t''$ it will also become inactive in $\I'$ at the same moment as a set becoming inactive only depends on the coloring duration of its subsets. This contradicts our assumption of the first event for $\I'$ occurring at $t'>t''$.
\end{itemize}
This shows that the invariant holds just before $t'$. We now show that events at $t'$ will not break this invariant. We note that multiple events may happen at the same moment, but as previously mentioned the order of considering the events does not change the final components. So, we assume that relevant events are taken in the same order in both runs and consider the effect of events at time $t'$ one at a time. There are again two cases for the event:
\begin{itemize}
    \item The event corresponds to an edge becoming fully colored. Let $\SET$ and $\SET'$ be the components in $\CC'$ containing the endpoints of this edge. If neither set contains $\Root$, then the same components contain these endpoints in $\CC$, and by the same argument as the previous case, this edge becomes fully colored at time $t'$ in the run for $\I$. So, we can add the edge in both runs, and the invariant will still hold for the new components. Otherwise, since the merged component will contain $\Root$, its addition to $\CC'$ does not affect the invariant and it will again hold.
    \item The event corresponds to a set $\SET$ becoming inactive. This set cannot contain $\Root$ as the set containing root has unlimited potential and never becomes inactive. So, by the same argument we used previously, this set must also be in $\CC$ and become inactive at the same time.
\end{itemize}

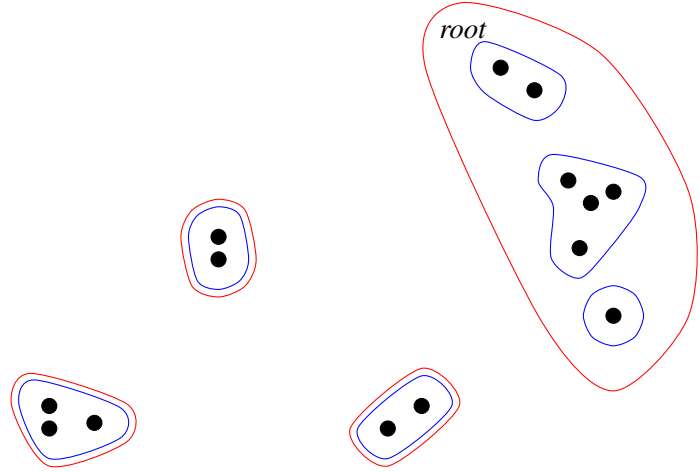
\begin{figure}
\begin{center}
\begin{tikzpicture}[scale=1]
\def\r{0.1}
\def\col{black}
\def\dx{1.5}
\def\dy{1.5}
\foreach \v/\x/\y in {
a1/0/0, a2/0/0.2, a3/0.4/0.05, 
b1/3/0, b2/3.3/0.2, 
c1/1.5/1.5, c2/1.5/1.7, 
d1/4/3.2, d2/4.3/3, 
e1/4.8/2, e2/5/2.1, e3/4.6/2.2, e4/4.7/1.6, 
f1/5/1} {
    \node[] (\v) at (\x*\dx,\y*\dy) {};
    \draw[fill=black] (\v) circle (\r);
};
\node at ($(d1) + (135:0.7)$) {\Root};
\def\mar{1}
\draw [red] plot [smooth cycle] coordinates {
($(d1) + (180:\mar)$) 
($(d1) + (120:\mar)$) 
($(d2) + (30:\mar)$) 
($(e2) + (0:\mar)$) 
($(f1) + (0:\mar)$) 
($(f1) + (270:\mar)$) 
($(f1) + (180:\mar)$) 
};
\def\mar{0.4}
\draw [blue] plot [smooth cycle] coordinates {
($(d1) + (240:\mar)$) 
($(d1) + (180:\mar)$) 
($(d1) + (120:\mar)$) 
($(d2) + (30:\mar)$) 
($(d2) + (330:\mar)$) 
($(d2) + (270:\mar)$) 
};
\def\mar{0.4}
\draw [blue] plot [smooth cycle] coordinates {
($(e3) + (240:\mar)$) 
($(e3) + (180:\mar)$) 
($(e3) + (120:\mar)$) 
($(e2) + (30:\mar)$) 
($(e2) + (330:\mar)$)
($(e4) + (0:\mar)$) 
($(e4) + (270:\mar)$) 
($(e4) + (200:\mar)$)
};
\def\mar{0.4}
\draw [blue] plot [smooth cycle] coordinates {
($(f1) + (0:\mar)$) 
($(f1) + (315:\mar)$) 
($(f1) + (270:\mar)$)
($(f1) + (225:\mar)$) 
($(f1) + (180:\mar)$) 
($(f1) + (135:\mar)$) 
($(f1) + (90:\mar)$)
($(f1) + (45:\mar)$) 
};

\foreach \mar/\col in {0.4/blue, 0.5/red} {
\draw [\col] plot [smooth cycle] coordinates {
($(c1) + (0:\mar)$) 
($(c1) + (315:\mar)$) 
($(c1) + (270:\mar)$)
($(c1) + (225:\mar)$) 
($(c2) + (180:\mar)$) 
($(c2) + (135:\mar)$) 
($(c2) + (90:\mar)$)
($(c2) + (45:\mar)$) 
};};

\foreach \mar/\col in {0.4/blue, 0.5/red} {
\draw [\col] plot [smooth cycle] coordinates {
($(a2) + (180:\mar)$) 
($(a2) + (120:\mar)$) 
($(a3) + (30:\mar)$) 
($(a3) + (330:\mar)$)
($(a1) + (270:\mar)$) 
};};

\foreach \mar/\col in {0.4/blue, 0.5/red} {
\draw [\col] plot [smooth cycle] coordinates {
($(b1) + (180:\mar)$) 
($(b1) + (270:\mar)$) 
($(b2) + (0:\mar)$) 
($(b2) + (90:\mar)$) 
};};
\end{tikzpicture}
\end{center}
\label{fig:inv}
\caption{An illustration of how the components in $\CC$ and $\CC'$ can be. The components in $\CC'$ are shown in red circles, and the components in $\CC$ are shown in blue ones. Each red component that does not include $\Root$ is also a blue component. }
\end{figure}
 This proves that the desired invariant will hold at all moments in the run for $\I'$. Now, let $\ys'$, $\ysv'$, and $\yv'$ denote the coloring duration values for this run and $\ys$, $\ysv$ and $\yv$ be the same values for the run on $\I$. Based on the above invariant, we will show that $\yv'\leq\yv$ for all vertices $v\neq\Root$.
 For any non-root vertex $v$, before it joins a component containing $\Root$ in the run on $\I'$, it belongs to the same component in both runs at any moment. In addition, these components will be active sets at the same moments.
 This is also true for all the vertices that are in the same component as $v$ in any of these moments. So, the $\ys$ and $\ys'$ values for these components up until this moment will be identical. Consequently, the $\ysv'$ values and therefore $\yv'$ will also be equal to their counterparts in the other run as the same ordering is used to assign coloring duration. After this moment, $\yv'$ will not increase anymore, as all coloring for the component will be assigned to $\Root$, and $\yv$ can only increase further. So, $\yv'\leq\yv$ for all $v\neq\Root$.
 We can also show that $y_{\Root}'\leq y_{\Root}$. Consider the moment the run on $\I$ ends. At this moment, the only component in $\CC$ is the one containing $\Root$. Based on the invariant, we can infer that this is also the only component in $\CC'$. So, the run on $\I'$ ends at least as soon as the run on $\I$. But as the component containing $\Root$ is always active and assigns its coloring to $\Root$, the $y$ value for the root is exactly the total duration of the process. So, $y_{\Root}'\leq y_{\Root}$.
 In addition, as any vertex $v\in\U$ can immediately merge with $\Root$ using the added $0$-weight edge, we will have $\yv'=0$ for these vertices.

 We can see from our proof of the invariant that any dead set in the run on $\I'$ will also be a dead set in the run on $\I$. Therefore, $\dead'\subseteq\dead$. This completes our proof of the lemma.
\end{proof}

\begin{algorithm}[ht]
  \caption{GW Algorithm}
  \label{alg:gw}
  \hspace*{\algorithmicindent} \textbf{Input:} undirected graph $\G=(\V, \E, \cc)$ with edge costs $\cc: \E \rightarrow \mathbb{R}_{\ge 0}$, root $\Root$, and penalties $\pi : \V \rightarrow \mathbb{R}_{\ge 0}$. \\
  \hspace*{\algorithmicindent} \textbf{Output:} Subtree $\T$ of $\G$ containing $\Root$, alongside a set $\dead$ of dead vertices.
  \begin{algorithmic}[1]
    \Procedure{\pcstgw}{$\I=(\G, \Root, \pi)$}
    \State Initialize $\F$ as an empty forest
    \State Initialize $\AS$ and $\CC$ as $\{\{v\}\mid v\in V\}$ \;
    \State Set $\ys\gets 0$ for all $\SET \in \AS$
    \State $\dead \gets \emptyset$
    \State $\DS \gets \emptyset$
    \While{$\left|\CC\right|>1$}
        \State $\Delta_1 \gets \min_{\SET\in\AS} (\sum_{v\in\SET}\pi(v) - \sum_{\SET'\subseteq\SET} y_{\SET'})$
        \State $\Delta_2 \gets \min_{e=(u,v)\in\E;\ e\cup\F\text{ is a forest}}(\frac{c_e-\sum_{\SET:e\in\delta(\SET)}\ys}{\left|\{\SET\in\AS\mid u\in \SET \vee v \in \SET\}\right|})$
        \State $\Delta \gets \min(\Delta_1,\Delta_2)$
        \For{$\SET \in \AS$}
        \State $\ys \gets \ys + \Delta$
        \EndFor
        \If{$\Delta_1 \leq \Delta_2$}
            \State Find a set $\SET$ minimizing $\Delta_1$
            \State $\AS \gets \AS - \{\SET\}$
            \State $\dead \gets \dead \cup \SET$
            \State $\DS \gets \DS \cup \{\SET\}$
        \Else
            \State Find an edge $e=(u,v)$ minimizing $\Delta_2$
            \State $\F \gets F\cup e$
            \State Update $\CC$ and $\AS$ accordingly
        \EndIf
    \EndWhile
    \State Extract $\T$ from $\F$ by repeatedly removing dead sets in $\DS$ that cut a single edge in $\F$
    \State \Return $(\T,\dead)$
    \EndProcedure
  \end{algorithmic}
\end{algorithm}


\section{The Iterative Algorithm}
\label{sec:iterative_alg}

In this section, we present our iterative algorithm which is described in Algorithm \ref{alg:rpcst}. In Section \ref{sec:iterative_analysis} we give an analysis for this algorithm.

Our algorithm makes use of the \pcstgw~procedure from Algorithm \ref{alg:gw} as a~fundamental component. Additionally, we employ an approximation algorithm for the~Steiner tree problem to improve the approximation factor. This can be any approximation algorithm for the Steiner tree problem. 
We denote the approximation factor for this algorithm as $\p$. Whenever we require this $\p$-approximation solution for the Steiner tree, we invoke the procedure named $\steinertree$. As our final approximation factor will depend on $\p$, we will use the current best approximation algorithm for Steiner Tree \cite{DBLP:conf/stoc/ByrkaGRS10} with $\p=\ln(4)+\epsilon$ in our analysis.
In addition, our algorithm depends on a constant $\bta$ which we will fix later in Section \ref{sec:find-alpha} to optimize the approximation ratio.

Our algorithm, as described in Algorithm~\ref{alg:rpcst}, identifies three solutions for the given PCST instance $\I=(\G,\Root,\pi)$. Subsequently, we opt for the solution with the minimum cost as the final solution for instance~$\I$.

First, we construct the instance $\Ib = (\G, \Root, \pb)$ from $\I$ by replacing $\pv$ with $\frac{\pv}{\bta}$ for all vertices.
One solution named ``\GW'' for instance $\I$, denoted as $\TGW$, can be obtained by invoking procedure~\pcstgw  (Line \ref{line:get_gw_output}) on instance $\Ib$, buying edges in $\TGW$ and paying penalties for vertices in $\NVGW$. From the definition of $\Ib$, we can conclude that $\pi(\NVGW) = \bta\pb(\NVGW)$. As stated in Section \ref{sec:gw}, in addition to $\TGW$, procedure~\pcstgw~also returns a set of vertices, $\DGW$, which represents dead vertices during the coloring process.

Another solution for instance $\I$ named ``\ST'' is obtained by retrieving a Steiner tree $\TST$ in graph $\G$ for the set of terminals $\NDGW \coloneqq \V \setminus \DGW$ which are the live vertices in the output of the $\GW$ algorithm. This solution is found using the procedure $\steinertree$ and is therefore a $\p$-approximation of the minimum Steiner tree on this terminal set. We pay the penalties for the vertices outside $\TST$, which will be a subset of $\DGW$.  


If $\DGW$ is empty, the algorithm immediately returns the solution with the lower total cost between the two obtained solutions. 
Otherwise, a third solution named ``\IT'', denoted as $\TIT$, is obtained through a recursive call on a simplified instance $\R$.
The simplified instance is formed through a process of adjusting penalties. 
We set the penalties for the vertices in $\DGW$, which are the dead vertices in the result of the $\pcstgw$ procedure, to zero while maintaining the penalty for the live vertices $\NDGW$, as indicated in Lines \ref{line:initial_pi} through \ref{line:construct_R}.


As a final step, the algorithm simply selects and returns the solution with the lowest cost. To help with the comparison of these three solutions, the algorithm calculates the values $\cGW=\cc(\TGW)+\pi(\NVGW)$, $\cST=c(\TST)+\pi(\NVST)$, and $\cIT=\cc(\TIT)+\pi(\NVIT)$, representing the costs of the solutions (as indicated in Lines \ref{line:costgw}, \ref{line:costst}, and \ref{line:costit}).

\begin{algorithm}[ht]
  \caption{Iterative PCST algorithm}
  \label{alg:rpcst}
  \hspace*{\algorithmicindent} \textbf{Input:} Undirected graph $\G=(\V, \E, \cc)$ with edge costs $\cc: \E \rightarrow \mathbb{R}_{\ge 0}$, root $\Root$, and penalties $\pi : \V \rightarrow \mathbb{R}_{\ge 0}$. \\
  \hspace*{\algorithmicindent} \textbf{Output:} Subtree $\T$ of $\G$ containing $\Root$.
  \begin{algorithmic}[1]
    \Procedure{\rpcst}{$\I=(\G, \Root, \pi)$}
        \State Construct $\pb$ by dividing all penalties by $\bta$. \label{line:set_pb}
        \State Construct the PCST instance $\Ib=(\G, \Root, \pb)$.
        \State $\TGW, \DGW \gets \pcstgw(\Ib)$
        \label{line:get_gw_output}
        \State $\cGW\gets c(\TGW)+\pi(\NVGW)$
        \label{line:costgw}
        \State $\NDGW\gets \{v: v \in \V, v \notin \DGW\}$ 
        \State $\TST\gets \steinertree(\G, \NDGW)$ 
        \label{line:get_st_output}
        \State $\cST\gets c(\TST)+\pi(\NVST)$
        \label{line:costst}
        \If{$\pi(\DGW)=0$} \label{line:check_Q_1_is_empty}
        \label{line:base-condition}
        \State\Return $\T_X$ where $\cost_X$ is minimum among $X \in \{\GW, \ST\}$
        \label{line:return-early}
        \EndIf
        \State Construct $\pi'$ by adjusting $\pi$ through the assignment of penalties for vertices in $\DGW$ to $0$.
        \label{line:initial_pi}
        \State Construct the PCST instance $\R=(\G, \Root, \pi')$.
        \label{line:construct_R}
        \State $\TIT \gets \rpcst(\R)$ \label{line:get_it_output}
        \label{line:recursive}
        \State $\cIT\gets c(\TIT)+\pi(\NVIT)$
        \label{line:costit}
        \State\Return $\T_X$ where $\cost_X$ is minimum among $X \in \{\GW, \ST, \IT\}$
        \label{line:return-min}
    \EndProcedure
  \end{algorithmic}
\end{algorithm}


\subsection{Analysis}
\label{sec:iterative_analysis}
For an arbitrary instance $\I = (\G, \Root, \pi)$ in PCST, our aim is to analyze the approximation factor achieved by Algorithm \ref{alg:rpcst}.
We compare the output of $\rpcst$ on $\I$ with an optimal solution $\OPT$ for the instance $\I$. 
We denote the tree selected in $\OPT$ as $\TOPT$. Then, the cost of $\OPT$ is given by $\cOPT=c(\TOPT)+\pi(\NVOPT)$.

We use an inductive approach to analyze the algorithm, where we focus on a single call of the algorithm and find upper bounds for each of our three solutions and a lower bound for the optimal solution $\OPT$. To find these lower and upper bounds, we make use of the coloring done by the \GW~algorithm on instance $\Ib$ and the values $\ys$, $\ysv$, and $\yv$ relating to this coloring process. In addition, we establish an upper bound for the solution obtained from the recursive call based on the induction hypothesis. 
In our inductive analysis, we only consider one individual call to the procedure at each time, to analyze either the induction base or the induction step. So, all the variables used in the analysis will relate to the algorithm's variables in the specific call we are analyzing. This includes the trees $\TGW$, $\TST$, and $\TIT$, and the live and dead vertices $\NDGW$ and $\DGW$.

We note that in our induction, we do not initially know the value of the approximation factor $\alf$ which we want to prove the algorithm achieves. Instead, we use $\alf$ as a variable in our inequalities, and this leads to a system of constraints involving $\alf$ that need to be satisfied for our induction to prove an $\alf$ approximation guarantee. 
These inequalities involve not only the approximation factor $\alf$ which we seek to find but also the parameter $\bta$ which defines the behavior of our algorithm. Throughout the analysis, we assume that $\bta\leq2$. 
We justify this assumption in Subsection \ref{sec:bta2} by showing that values of $\bta>2$ cannot lead to a better than $2$ approximation. 
To determine our approximation factor $\alf$, we consider the range $\p \leq \alf \leq 2$. 
This range is chosen because we cannot assume that our algorithm performs better than the Steiner tree algorithm, which we use as a component. 
Additionally, our solution is guaranteed to be at least as good as the $2$-approximation provided by the \GW~algorithm.

In the first step, we categorize non-root vertices based on the output of $\pcstgw(\Ib)$ and $\OPT$. This categorization helps us establish more precise bounds for the solutions by enabling a more tailored analysis within each category.

\begin{definition}
\label{def:tbl}
For an instance $\I$, $\OPT$ partition vertices into two sets: $\V(\TOPT)$ and $\NVOPT$. $\pcstgw(\Ib)$ also partitions vertices into two sets: $\NDGW$ and $\DGW$. We define four sets to categorize the vertices, excluding $\Root$, based on these two partitions:
\begin{align*}
    &\A = \V(\TOPT) \cap \NDGW - \RSET
    &\B = \V(\TOPT) \cap \DGW\\
    &\C = \NVOPT \cap \NDGW
    &\D = \NVOPT \cap \DGW\\
\end{align*}
\begin{table}[H]
\begin{center}
\begin{tabular}{|c|c|c|c|}
\cline{3-4}
\multicolumn{2}{c|}{} & \multicolumn{2}{c|}{\pcstgw(\Ib)} \\
\cline{3-4}
\multicolumn{2}{c|}{} & Live vertices\footnotemark & Dead vertices \\
\hline
\multirow{2}{*}{Optimal Solution} & Connected to \Root \footref{fnt1}
& \A & \B \\
\cline{2-4}
& Penalty paid& \C & \D \\
\hline
\end{tabular}
\caption{This table illustrates the categories of vertices.}
\end{center}
\end{table}
\footnotetext{\label{fnt1}excluding \Root.}

Using the coloring scheme of $\pcstgw(\Ib)$, we introduce the following values to represent the total duration of coloring with vertices in these sets.
\begin{align*}
    &\va = \sum_{v \in \A} \yv
    &\vb = \sum_{v \in \B} \yv\\
    &\vc = \sum_{v \in \C} \yv
    &\vd = \sum_{v \in \D} \yv\\
\end{align*}
\end{definition}

\begin{definition}[Connected and unconnected dead vertices]
\label{def:prime-doubleprime}
    For an instance $\I$, based on Definition~\ref{def:tbl}, the sets $\B$ and $\D$ represent dead vertices in the output of $\pcstgw(\Ib)$. 
    We further divide set $\B$ into $\Bp$ and $\Bz$, and set $\D$ into $\Dp$ and $\Dz$, based on whether they are connected to the root at the end of the $\pcstgw(\Ib)$ procedure.
    Let $\Bp$ and $\Dp$ be the subsets of $\B$ and $\D$, respectively, representing the vertices connected to the root. Similarly, $\Bz$ and $\Dz$ are the subsets of $\B$ and $\D$, respectively, indicating the vertices not connected to the root at the end of the procedure.
\begin{align*}
    &\Bp = \B \cap \V(\TGW) = \V(\TOPT) \cap \DGW \cap \V(\TGW)\\
    &\Bz = \B \cap \NVGW = \V(\TOPT) \cap \DGW \cap \NVGW\\
    &\Dp = \D \cap \V(\TGW) = \NVOPT \cap \DGW \cap \V(\TGW)\\
    &\Dz = \D \cap \NVGW) = \NVOPT \cap \DGW \cap \NVGW\\
\end{align*}
    Subsequently, we define $\vbp$, $\vbz$, $\vdp$, and $\vdz$ as the total duration of coloring with vertices in sets $\Bp$, $\Bz$, $\Dp$, and $\Dz$, respectively.
\begin{align*}
    &\vbp = \sum_{v \in \Bp} \yv
    &\vbz = \sum_{v \in \Bz} \yv\\
    &\vdp = \sum_{v \in \Dp} \yv
    &\vdz = \sum_{v \in \Dz} \yv\\
\end{align*}
It is trivial to see that $\vd=\vdp+\vdz$ as $\Dp\cup\Dz=\D$ and $\Dp\cap\Dz=\emptyset$. Similarly, $\vb=\vbp+\vbz$.
\end{definition}

\begin{definition}[Single-edge and multi-edge sets]
\label{def:single_multi_edge}
    For an instance $\I$, we call a set $S \subseteq \V$ a single-edge set if $|\delta(S) \cap \TOPT| = 1$ and a multi-edge set if $|\delta(S) \cap \TOPT| > 1$ (Illustrated in Figure \ref{fig:single-multi-cut}).
    We assign each moment of coloring with colors of vertices in $\B$ which are inside a single-edge set to $\vba$, and those in a multi-edge set to $\vbb$.
    These definitions are as follows:
\begin{align*}
    \vba =\sum_{v \in \B} \sum_{|\delta(S) \cap \TOPT| = 1} \ysv\\
    \vbb =\sum_{v \in \B} \sum_{|\delta(S) \cap \TOPT| > 1} \ysv
\end{align*}
    Note that $\vb = \vba + \vbb$, as every vertex in $\B$ is connected to $\Root$ in the optimal solution. 
    Therefore, with each moment of coloring involving vertices in $\B$, the corresponding active set cuts an edge belonging to the path from that vertex to $\Root$ in the optimal solution.
\begin{figure}
\begin{center}
\begin{tikzpicture}[scale=1.3]
\def\r{0.1}
\foreach \v/\x/\y/\dir/\l/\col in {0/0/0/above/$r$/black, 1/-1/-1/0//black, 2/1/-1/0//red, 3/-2.1/-1/below/Multi-edge/blue, 4/0/-2/0//black} {
    \node[label={[text=\col]\dir:\l}] (\v) at (\x, \y) {};
    \draw[\col, fill=\col] (\v) circle (\r);
};
\draw (1) -- (0) -- (2);
\draw (1) -- (4);
\def\col{blue}
\draw[\col, line width=3] (3) circle (1.6);
\draw[\col, line width=3] ($(1)!0.13!(3)$) -- ($(3)!0.13!(1)$);
\draw[\col, line width=3] (1) -- ($(1)!0.45!(0)$);
\draw[\col, line width=3] (1) -- ($(1)!0.45!(4)$);
\def\col{red}
\draw[\col, line width=3] (2) circle (0.7);
\draw[\col, line width=3] (2) -- ($(2)!0.5!(0)$);
\node[label={[text=\col]below:Single-edge}] at ($(2) - (0, 0.7)$) {};
\end{tikzpicture}
\end{center}
\caption{Illustration of single-edge set vs. multi-edge set in $\TOPT$. The red set is a single-edge set, but the blue one is a multi-edge set.}
\label{fig:single-multi-cut}
\end{figure}
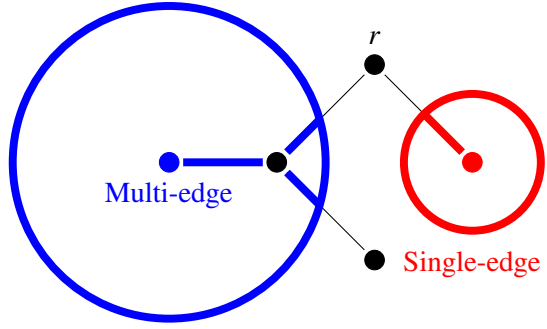
\end{definition}

\begin{lemma}
\label{lm:gwdead_beta}
     For any vertex $v \in \V$, we have $\bta \yv \leq \pi(v)$. Furthermore, if $v \in \B \cup \D$, which means is a dead vertex, we have $\bta \yv=\pi(v)$.
\end{lemma}
\begin{proof}
    Since we run $\pcstgw$ on $\pb$ in Line~\ref{line:get_gw_output}, we can use Lemma~\ref{lemma:gwdead} using penalties $\pb$.
    That means, for any vertex $v \in \V$, we have $\yv \leq \pb(v)$, and if $v$ is a dead vertex, we have $\yv=\pb(v)$.
    Since in Line~\ref{line:set_pb}, we set $\pb(v) = \dfrac{\pi(v)}{\bta}$, we can conclude the lemma.
\end{proof}

Now for a given instance $\I$, we derive lower bounds on the optimal solution using terms defined earlier. We use a similar approach that is used in~\cite{DBLP:journals/corr/abs-2309-05172} to bound the optimal solution.

\begin{lemma}
\label{lm:opt_basic_bound}
    We can bound the cost of the optimal solution in terms of the cost of its tree as follows:
    $$\cOPT \ge c(\TOPT) + \bta \vc + \bta \vd \text{.}$$
\end{lemma}
\begin{proof}
    According to the definition of cost in PCST, we can determine the cost of the optimal solution by separately calculating the weight of its tree and the penalties it pays. Additionally, based on Definition \ref{def:tbl}, we have $\NVOPT = \C \cup \D$. Utilizing these two observations, we can establish an upper bound for $\cOPT$ as follows:
\begin{align*}
    \cOPT &= c(\TOPT) + \pi(\NVOPT)\\
    &= c(\TOPT) + \sum_{v \in \NVOPT} \pi(v)\\
    &= c(\TOPT) + \sum_{v \in \C} \pi(v) + \sum_{v \in \D} \pi(v)\tag{$\C \cap \D = \emptyset$}\\
    &\ge c(\TOPT) + \sum_{v \in \C} \bta\yv + \sum_{v \in \D} \bta\yv\tag{Lemma \ref{lm:gwdead_beta}}\\
    &= c(\TOPT) + \bta\vc + \bta\vd \tag{Definition \ref{def:tbl}} 
\end{align*}
\end{proof}

Based on Lemma \ref{lm:opt_basic_bound}, we can easily conclude the following corollary which bounds the weight of the optimal solution tree using the cost of the optimal solution.

\begin{corollary}
\label{col:bound_T_OPT}
    We can bound the cost of optimal solution's tree as follows:
    $$c(\TOPT) \le \cOPT - \bta \vc - \bta \vd\text{.}$$
\end{corollary}

Now we use Lemma \ref{lm:opt_basic_bound}, to expand the bound of the optimal solution.

\begin{lemma}
\label{lm:opt_main_bound}
    We can establish a lower bound for the optimal solution as follows:
    $$\cOPT \ge \va + \vba + 2\vbb + \bta \vc + \bta \vd$$
\end{lemma}
\begin{proof}
First, we demonstrate that $\va + b_1 + 2b_2$ is a lower bound for $c(\TOPT)$. To achieve this, for any set $\SET$, we define $\dopt(\SET)$ as the number of edges of $\TOPT$ that are colored by $\SET$.
Given that each portion of an edge will be colored at most once, and each set $\SET \subseteq \V$ colors $\dopt(\SET)\cdot\ys$ of the optimal solution, we can derive a lower bound for $\cc(\TOPT)$ based on the proportion of the colored edges in $\TOPT$.
\begin{align*}
    \cc(\TOPT) &\ge \sum_{\SET \subseteq \V} \dopt(\SET)\cdot\ys\\
    &= \sum_{\SET \subseteq \V} \sum_{v \in S} \dopt(S)\cdot\ysv \tag{$\ys = \sum_{v \in S}\ysv$}\\ 
    &= \sum_{\substack{v \in \V}} \sum_{\substack{\SET \subseteq \V\\ v \in S}}  \dopt(S)\cdot\ysv \tag{change the order of summations}\\
    &\ge \sum_{\substack{v \in \A \cup \B}} \sum_{\substack{\SET \subseteq \V\\v \in S}}  \dopt(S)\cdot\ysv \tag{$\A \cap \B = \emptyset$, $\A\cup\B \subseteq \V$}\\
    &\ge\sum_{v\in\A} \sum_{\substack{\SET \subseteq \V\\v \in S}} \dopt(S)\cdot\ysv + \sum_{v\in\B} \sum_{\substack{\SET \subseteq \V\\v \in S}} \dopt(S)\cdot\ysv\text{.} 
\end{align*}
Furthermore, for any vertex $v$ in $\A$ or $\B$, based on Definition~\ref{def:tbl}, there exists a path from $v$ to $\Root$ in $\TOPT$.
Also, for every set $S \subseteq \V$ where $\ysv > 0$, we know $\Root \notin S$ otherwise all coloring of set $S$ would be assigned to $\Root$.
Using these two observations, we can infer that at least one edge of $\TOPT$ is colored by $S$, resulting in $\dopt(S) \ge 1$.

For vertices in $\A$, we have:
\begin{align*}
    \sum_{v\in\A} \sum_{\substack{\SET \subseteq \V\\v \in S}} \dopt(S)\cdot\ysv &\ge 
    \sum_{v\in\A} \sum_{\substack{\SET \subseteq \V\\v \in S}} \ysv \tag{$\dopt(S)\geq 1$ when $\ysv>0$}\\
    &= \sum_{v\in \A} \yv \tag{Definition \ref{def:coldur}}\\
    &= \va \tag{Definition \ref{def:tbl}}.
\end{align*}

For vertices in $\B$, we have:
\begin{align*}
    \sum_{v\in\B} \sum_{\substack{\SET \subseteq \V\\v \in S}} \dopt(S)\cdot\ysv &=\sum_{v\in\B} \sum_{\substack{\SET \subseteq \V\\v \in S\\ \dopt(S)=1}} \dopt(S)\cdot\ysv + \sum_{v\in\B} \sum_{\substack{\SET \subseteq \V\\v \in S\\ \dopt(S)>1}} \dopt(S)\cdot\ysv\tag{$\dopt(S)\geq 1$ when $\ysv>0$}\\
    &\ge\sum_{v \in \B} \sum_{\substack{\SET \subseteq \V\\v \in S\\ \dopt(S)=1}} \ysv + \sum_{v \in \B} \sum_{\substack{\SET \subseteq \V\\v \in S\\ \dopt(S)>1}} 2\ysv \\
    &= \vba + 2\vbb\text{.} \tag{Definition \ref{def:single_multi_edge}}
\end{align*}
Combining all together, we obtain:
 \begin{align*}
    \cc(\TOPT) &\ge \va + \vba + 2\vbb\text{.}
\end{align*}
By using this bound along with Lemma \ref{lm:opt_basic_bound}, we can bound $\cOPT$.
\begin{align*}
    \cOPT &\ge c(\TOPT) + \bta \vc + \bta \vd \tag{Lemma \ref{lm:opt_basic_bound}}\\
    &\ge \va + \vba + 2\vbb + \bta \vc + \bta \vd\text{.} \qedhere
\end{align*}
\end{proof}
Next, we bound the $\GW$ solution.
\begin{lemma}
\label{lm:gw_basic_bound}
The following bound holds for the cost of the solution returned by the output of $\pcstgw(\Ib)$ for instance $\I$:
    $$\cGW \le 2\va + 2\vb + 2\vc + 2\vd\text{.}$$
\end{lemma}
\begin{proof}
    According to Line \ref{line:costgw} of Algorithm \ref{alg:rpcst}, we have
\begin{align*}
     \cGW = c(\TGW)+\pi(\NVGW)
\end{align*}
    To start, based on Definition \ref{def:tbl} we have
\begin{align*}
    \A \cup \C = (\V(\TOPT) \cap \NDGW - \RSET) \cup (\NVOPT \cap \NDGW)
    = \NDGW - \RSET\text{,}
\end{align*}
    and based on Definition \ref{def:prime-doubleprime} we have
\begin{align*}
    \Bp \cup \Dp = (\V(\TOPT) \cap \DGW \cap \V(\TGW)) \cup (\NVOPT \cap \DGW \cap \V(\TGW)) = \DGW \cap \V(\TGW)\text{.}
\end{align*}
Then, we can combine them to obtain
\begin{align*}
    \A \cup \C \cup \Bp \cup \Dp 
    &= (\NDGW - \RSET) \cup (\DGW \cap \V(\TGW))\\
    &= ((\NDGW \cap \V(\TGW)) - \RSET) \cup (\DGW \cap \V(\TGW))\tag{$\NDGW \subseteq \V(\TGW)$}\\
    &= \V(\TGW) - \RSET\text{.}
\end{align*}
    Applying this observation to Lemma \ref{lemma:gwtree} results in a bound for $c(\TGW)$.
\begin{align*}
    \cc(\TGW) &\le 2\cdot\sum_{\substack{v \in \V(\TGW)-\RSET}}\yv\\
    &= 2\cdot\sum_{v \in \A \cup \C \cup \Bp \cup \Dp}\yv\\
    &\le 2\cdot\sum_{v \in \A} \yv
    + 2\cdot\sum_{v \in \C} \yv
    + 2\cdot\sum_{v \in \Bp} \yv
    + 2\cdot\sum_{v \in \Dp} \yv\\
    &= 2\va + 2\vc + 2\vbp + 2\vdp
    \text{.} \tag{Definitions \ref{def:tbl} and \ref{def:prime-doubleprime}}
\end{align*}
    Additionally, in GW, we pay penalties for the vertices that are not connected to the $\Root$, all of which are dead according to Lemma \ref{lemma:gwlive}. Consequently, we can deduce that:
\begin{align*}
     \pi(\NVGW) &= \sum_{v \notin \V(\TGW)} \pi(v)\\
     &= \sum_{v \notin \V(\TGW)} \bta\yv\tag{Lemma \ref{lm:gwdead_beta}}\\
     &= \sum_{v \in \Bz} \bta\yv + \sum_{v \in \Dz} \bta\yv
     = \bta\vbz + \bta\vdz\tag{Definition \ref{def:prime-doubleprime}}\text{.}
\end{align*}
It is worth emphasizing that throughout the algorithm, we assume $\bta \leq 2$.
In conclusion, we can establish an upper bound for $\cGW$.
\begin{align*}
    \cGW &= c(\TGW)+\pi(\NVGW)\\
    & \le 2\va + 2\vc + 2\vbp + 2\vdp + \bta\vbz + \bta\vdz\\
    & \le 2\va + 2\vc + 2\vbp + 2\vbz +  2\vdp + 2\vdz\tag{$\bta \le 2$}\\
    & = 2\va + 2\vb + 2\vc + 2\vd\tag{Definition \ref{def:prime-doubleprime}}\text{.}
\end{align*}
\end{proof}

We restate this upper bound in terms of the variable $\alf$ and the cost of the optimal solution $\cOPT$ using Lemma \ref{lm:opt_main_bound}.
\begin{lemma}
The following bound holds for the cost of the solution returned by the output of $\pcstgw(\Ib)$ for instance $\I$:
    $$\cGW \le \alf \cdot \cOPT + (2-\alf) \va + (2-\alf) \vba + (2 - 2\alf) \vbb + (2 - \alf\bta) \vc + (2 - \alf \bta) \vd\text{.}$$
\end{lemma}
\begin{proof}
We can directly apply Lemma \ref{lm:opt_main_bound} to the previous bound obtained in the preceding Lemma \ref{lm:gw_basic_bound}.
    \begin{align*}
        \cGW &\le 2\va + 2\vb + 2\vc + 2\vd \tag{Lemma \ref{lm:gw_basic_bound}}\\
        &\le 2\va + 2(\vba + \vbb) + 2\vc + 2\vd + \alpha \cdot (\cOPT - \va - \vba - 2\vbb - \bta \vc - \bta \vd) \tag{Lemma \ref{lm:opt_main_bound}}\\
        &\le \alf \cdot \cOPT + (2-\alf) \va + (2-\alf) \vba + (2 - 2\alf) \vbb + (2 - \alf\bta) \vc + (2 - \alf \bta) \vd\text{.} 
    \end{align*}
\end{proof}

Next, we bound the cost of the \ST~solution.
For a set $\SET$, let $\T_{\OPTp_\SET}$ denote the minimum cost Steiner tree on this set. In the following lemma, we relate the cost of the $\ST$ solution to the cost of $\T_{\OPTp_\NDGW}$. 
\begin{lemma}
\label{lm:st_basic_bound}
For instance $\I$, we can bound the cost of the solution returned by the output of $\ST$ as follows:
    $$\cST \le \p  \cdot c(\T_{\OPTp_{\NDGW}}) + \bta \vb + \bta \vd\text{.}$$
\end{lemma}
\begin{proof}
    Since in $\TST$, we are connecting every vertex in $\NDGW$ to $\Root$, using an Steiner tree algorithm with an approximation factor of $\p$, the cost of the tree $\TST$ can be bounded by 
    $$c(\TST) \le \p \cdot c(\T_{\OPTp_{\NDGW}})\text{.}$$

    Moreover, as all vertices in $\NDGW$ are connected to $\Root$, the vertices for which we need to pay penalties for this solution form a subset of $\DGW$, i.e., $\NVST \subseteq \DGW$.
    Furthermore, by Definition \ref{def:tbl} we have:
\begin{align*}
    \B \cup \D 
    &= (\V(\TOPT) \cap \DGW) \cup (\NVOPT \cap \DGW)\tag{Definition \ref{def:tbl}}\\
    &= \DGW
\end{align*}
    Now, we can bound the penalty paid by the $\ST$ solution.
\begin{align*}
    \pi(\NVST) &\le \pi(\DGW)\\
    &= \pi(\B \cup \D)\\
    &= \sum_{v \in \B \cup \D} \pi(v)\\
    &= \sum_{v \in \B \cup \D} \bta\yv\tag{Lemma~\ref{lm:gwdead_beta}}\\
    &\le \sum_{v \in \B} \bta\yv + \sum_{v \in \D} \bta\yv \\
    &= \bta \vb + \bta \vd \tag{Definition~\ref{def:tbl}}
\end{align*}
    Finally, we use these bounds to complete the proof
    $$\cST = c(\TST) + \pi(\NVST) \le \p \cdot c(\T_{\OPTp_{\NDGW}}) + \bta \vb + \bta \vd$$
\end{proof}

We now provide an upper bound for the cost of $\T_{\OPTp_{\NDGW}}$ based on the cost of $\TOPT$ to obtain our main upper bound for \ST.
\begin{lemma}
\label{lm:add_C_to_AB_tree}
For the minimum cost Steiner tree $\T_{\OPTp_{\NDGW}}$ on $\NDGW$, we have
    $$c(\T_{\OPTp_{\NDGW}}) \le c(\TOPT) + 2 \vc + 2 \vd.$$
\end{lemma}
\begin{proof}
We construct a new instance $\Ib'=(\G', \Root, \pb)$ where $\G'$ is obtained from $\G$ by adding a set $\E_0$ of edges of weight $0$ from $\Root$ to every vertex in $\U=\A\cup\B=\V(\TOPT)-\RSET$. Let $\TGW'$ be the resulting tree and $\yv'$ be the coloring duration for the vertices in this process assuming we assign the colors in the same way as we did when running the \GW~algorithm on $\Ib$. By Lemma \ref{lem:gwlocal}, $\yv'\leq\yv$ for all vertices in $\C\cup\D$. In addition, we have $\yv'=0$ for all vertices in $\U=\A\cup\B$. Then, using Lemma \ref{lemma:gwtree} we can bound the cost of $\TGW'$ as
 \begin{align*}
     \cc(\TGW') &\leq 2\sum_{v\in\V(\TGW')-\RSET}\yv' \tag{Lemma \ref{lemma:gwtree}}\\
     &\leq 2\sum_{v\in\V-\RSET}\yv' \tag{$\V(\TGW')\subseteq\V$}\\
     &\leq 2\sum_{v\in\A\cup\B}\yv' + 2\sum_{v\in\C\cup\D}\yv' \tag{$\A\cup\B\cup\C\cup\D=\V-\RSET$}\\
     &\leq 2\sum_{v\in\C\cup\D}\yv' \tag{$\yv'=0$ if $v\in\A\cup\B$ by Lemma \ref{lem:gwlocal}}\\
     &\leq 2\sum_{v\in\C\cup\D}\yv \tag{$\yv'\leq\yv$ by Lemma \ref{lem:gwlocal}}\\
     &=2\vc + 2\vd. \tag{Definition \ref{def:tbl}}
 \end{align*}

Let $\dead'$ be the set of dead vertices returned by the \GW~algorithm on $\Ib'$. Based on Lemma \ref{lem:gwlocal}, we have $\dead'\subseteq\DGW$. Therefore, as vertices in $\A\cup\C\cup\RSET=\NDGW$ are not part of $\DGW$, they cannot be part of $\dead'$ either and must be live vertices in this run. Lemma \ref{lemma:gwlive} means that these vertices are connected by $\TGW'$.

If we remove any edges in $\E_0$ from $\TGW'$, and instead add $\TOPT$, which is a spanning tree on $\A\cup\B\cup\RSET$, all the vertices in $\V(\TGW')$ will remain connected. So, we get a connected subgraph of $\G$ that connects $\NDGW$. The cost of this subgraph is at most 
\begin{align*}
    \cc((\TGW'-\E_0)\cup\TOPT) &\leq \cc(\TOPT)+\cc(\TGW')\\
    &\leq \cc(\TOPT) + 2\vc + 2\vd.
\end{align*}
As this subgraph connects $\NDGW$, its cost gives us an upper bound on the cost of the minimum Steiner tree on these vertices. So we have 
    $$\cc(\T_{\OPTp_{\NDGW}}) \le \cc(\TOPT) + 2 \vc + 2 \vd.$$

\end{proof}

We combine the last two lemmas to introduce an upper bound for the \ST~solution. We again state this upper bound in terms of $\cOPT$ and $\alf$. Here, we rely on the fact that $\alf\geq\p$ to add a non-negative value to an initial upper bound based on Lemmas \ref{lm:st_basic_bound} and \ref{lm:add_C_to_AB_tree}.
\begin{lemma}
For instance $\I$, we can bound the cost of the solution returned by the output of $\ST$ as follows:
    $$\cST \le \alf \cdot \cOPT + (\p - \alf) \va + (\p + \bta - \alf) \vba + (2\p + \bta - 2\alf) \vbb + (2\p - \alf \bta) \vc + (2\p + \bta - \alf \bta) \vd\text{.}$$
\end{lemma}
\begin{proof}
By combining Lemma \ref{lm:st_basic_bound} with Lemma \ref{lm:add_C_to_AB_tree}, we can derive a new bound for $\cST$.
\begin{align*}
    \cST \leq{}& \p \cdot c(\T_{\OPTp_{\NDGW}}) + \bta \vb + \bta \vd \tag{Lemma \ref{lm:st_basic_bound}}\\
    \leq{}&\p(c(\TOPT) + 2 \vc + 2 \vd) + \bta \vb + \bta \vd \tag{Lemma \ref{lm:add_C_to_AB_tree}}\\
    \leq{}&\p(\cOPT - \bta \vc - \bta \vd + 2 \vc + 2 \vd) + \bta \vb + \bta \vd \tag{Corollary~\ref{col:bound_T_OPT}}\\
    \leq{}&\p(\cOPT - \bta \vc - \bta \vd + 2 \vc + 2 \vd) + \bta \vb + \bta \vd  \\&+ (\alf - \p) (\cOPT - \va - \vba - 2\vbb - \bta \vc - \bta \vd) \tag{Lemma \ref{lm:opt_main_bound}, $\alf-\p\geq0$}\\
    ={}&\alf \cdot \cOPT + (\p - \alf) \va + (\p + \bta - \alf) \vba + (2\p + \bta - 2\alf) \vbb + (2\p - \alf \bta) \vc + (2\p + \bta - \alf \bta) \vd 
\end{align*}
\end{proof}

Now, assume that we want to show that the algorithm achieves an approximation factor of $\alf$. Then, to prove this by induction, we need to show two things. 
First, we need to show that in the base case where the dead set $\DGW$ returned by the \GW~algorithm has penalty $0$ and we do not make a recursive call, our solution is an $\alf$ approximation.
Secondly, we have to demonstrate the induction step. 
This means that we have to show that if our recursive call on instance $\R$ returns an $\alf$ approximation for this instance, the final returned solution will also be an $\alf$ approximation. If these two steps are accomplished, then by induction on the number of vertices with non-zero penalties (which decreases with every recursive call), we can prove that our algorithm achieves an $\alf$ approximation.

So far, we do not know the value of $\alf$ so we cannot prove the induction steps directly. Instead, we will show that if $\alf$ satisfies certain constraints then both the base case and the step of induction can be proven for that value of $\alf$ and therefore our algorithm will give us an $\alf$ approximation.
These constraints are obtained by thinking of $\alf$ as a variable and then trying to prove the induction base and the induction step for $\alf$. Minimizing $\alf$ in this system of constraints will give us an upper bound on the approximation factor of our algorithm.

In the following, we first assume that the recursive call on $\R$ is an $\alf$ approximation, and bound the iterative solution using this assumption. Then, in Section \ref{sec:find-alpha} we combine the bounds for the different solutions to find a system of constraints that restrict $\alf$. We also consider the constraints that arise from the base case being an $\alf$ approximation, which turn out to form a subset of the former constraints. Finally, we find the minimum value of $\alf$ that can satisfy these constraints to obtain our approximation guarantee. 


We start with the next lemma, which bounds the cost of the iterative solution's output, assuming that the recursive call returns an $\alf$ approximate solution for instance $\R$.
Here, $\OPT_{\R}$ denotes the optimal solution for the PCST instance $\R$.
\begin{lemma}
\label{lm:it_basic_bound}
For instance $\I$, the cost of the iterative solution, denoted as $\cIT$, can be bounded as follows:
    $$\cIT \le \alf\cdot \cOPTR + \bta \vb + \bta \vd\text{,}$$
assuming that the recursive call on instance $\R$ returns an $\alf$ approximate solution.
\end{lemma}
\begin{proof}
    Based on our assumption, $\rpcst(\R)$ will return a solution that is an $\alf$-approximate of the optimal solution of instance $\R$ which we indicate by $\OPT_{\R}$.
    This gives us the following bound:
    $$c(\TIT) + \pi'(\NVIT) \le \alf \cdot \cOPTR\text{.}$$
    However, as $\cIT = c(\TIT) + \pi(\NVIT)$, we need to establish the relationship between $\pi(\NVIT)$ and $\pi'(\NVIT)$. The only difference between these functions lies in setting the penalty for vertices in $\DGW = \B \cup \D$ to zero in $\pi'$, as indicated in Line~\ref{line:initial_pi}. Thus, we can conclude that
    \begin{align*}
        \pi(\NVIT) &\le \pi'(\NVIT) + \pi(\B \cup \D) \\
        &= \pi'(\NVIT) + \bta\sum_{v \in (\B \cup \D)} \yv \tag{Lemma~\ref{lm:gwdead_beta}}\\
        &= \pi'(\NVIT) + \bta\vb + \bta\vd\text{.} \tag{Definition~\ref{def:tbl}}
    \end{align*}
    By combining these inequalities, we get
    $$\cIT = c(\TIT) + \pi(\NVIT) \le c(\TIT) + \pi'(\NVIT) + \bta\vb + \bta\vd \le \alf \cdot\cOPTR + \bta \vb + \bta \vd\text{.}$$
\end{proof}

\begin{lemma}
\label{lm:remove_b_1_edges}
    For an instance $\I$, we can remove a set of edges with a total length of $\vba$ from $\TOPT$ in such a way that the vertices in $\A$ remain connected to $\Root$.
\end{lemma}
\begin{proof}
    Consider a moment of coloring with the color of a vertex $v \in \B$ in a single-edge set $S \subseteq \V$.
    Given that we are coloring with $v$ at this moment, the vertex is still a live vertex.
    However, since $v$ is in $\B$, it will become dead at some moment of the algorithm.
    Since all the vertices in $S$ will remain in the same component until the end of the algorithm, the moment $v$ becomes dead, all vertices in $S$ will also become dead.
    That means, every vertex in $S$ is either in $\B$ or $\D$, i.e. $S \subseteq \B \cup \D=\DGW$.

    Since $S$ is a single-edge set, there is only one edge from $\TOPT$ that cuts this set.
    Let assume that this edge is $e$, i.e. $\delta(S) \cap \TOPT = \{e\}$.
    Removing edge $e$ from $\TOPT$, will only disconnect vertices in $S$ from $\Root$, since $S$ is a single-edge set and paths in $\TOPT$ from $\Root$ to vertices outside of $S$ will not pass through $e$.
    
    If we remove all such edges from $\TOPT$, the total cost of the removed edges will be at least $\vba$. This is due to the fact that the coloring on these edges from single-cut sets assigned to the vertices in $\B$ is equal to $\vba$, and the coloring on each edge is at most its weight. 
    Note that, each single-edge set is coloring exactly one edge of the optimal solution at each moment. So, we can remove edges with a total length of at least $\vba$ from $\TOPT$ without disconnecting vertices in $\A$ from $\Root$.
\end{proof}

\begin{lemma}
\label{lm:remaining_bound}
    For an instance $\I$, we can bound the cost of the optimal solution for instance $\R$ by 
    $$\cOPTR \le \cOPT - \bta \vd - \vba\text{,}$$
    where $\R$ is created at Line~\ref{line:construct_R} of $\rpcst(\I)$.
\end{lemma}
\begin{proof}
    To prove this lemma, we start by showing that there is a solution for instance $\R$ that costs at most $\cOPT - \bta \vd - \vba$.
    Since $\OPT_{\R}$ is the optimal solution of instance $\R$, its cost would not exceed the cost of the instance we are constructing. 
    This will complete the proof of the lemma.
    To construct the mentioned instance, we take the optimal solution of instance $\I$, which we indicate by $\OPT$, and remove extra edges from its tree $\TOPT$. 
    Additionally, we do not need to pay penalties for pairs in $\OPT$ whose penalty is set to zero in $\pi'$ at Line \ref{line:initial_pi} for instance $\R$.

    Let's start with the tree $\TOPT$.
    Using Lemma \ref{lm:remove_b_1_edges}, we can remove a set of edges from $\TOPT$ with a total length of at least $\vba$ without disconnecting vertices in set $\A$ from $\Root$.

    Moreover, the optimal solution pays penalties for vertices in set $\C \cup \D$. 
    However, instance $\R$ has been constructed by assigning zero to the penalty of vertices in set $\DGW$, which includes vertices in set $\D$. 
    Therefore, the penalty that we pay for vertices in $\D$ in the optimal solution is not required to be paid in $\OPT_{\R}$. 
    This deducts $\pi(\D)$ from the cost of the optimal solution, which is equal to $\bta \vd$ according to Lemma~\ref{lm:gwdead_beta}. 
    This completes the proof of this lemma.
\end{proof}

\begin{lemma}
For instance $\I$, the output of the iterative solution can be bounded as follows:
    $$\cIT \le \alf\cdot \cOPT + (\bta - \alf) \vba + \bta \vbb + (\bta - \alf \bta) \vd$$
assuming that the recursive call on instance $\R$ returns an $\alf$ approximate solution.
\end{lemma}
\begin{proof}
We utilize Lemma \ref{lm:remaining_bound} to modify the terms of the bound in Lemma \ref{lm:it_basic_bound}.
\begin{align*}
    \cIT &\le \alf \cOPTR + \bta \vb + \bta \vd \tag{Lemma \ref{lm:it_basic_bound}}\\
    &\le \alf (\cOPT - \bta \vd - \vba) + \bta (\vba + \vbb) + \bta \vd \tag{Lemma \ref{lm:remaining_bound} and Definition \ref{def:single_multi_edge}}\\
    &\le \alf \cdot\cOPT + (\bta - \alf) \vba + \bta \vbb + (\bta - \alf \bta) \vd 
\end{align*}
\end{proof}

\subsection{Finding The Approximation Factor}
\label{sec:find-alpha}

Now that we have bounded $\cGW$, $\cST$, and $\cIT$, we can determine an appropriate value for $\alf$ such that, during each call of $\rpcst$ on instance $\I$, the minimum of $\cGW$, $\cST$, and $\cIT$ is at most $\alf \cdot \cOPT$.

To achieve this, we assign weights to each solution in a way that the weighted average of these three bounds is at most $\alf \cdot \cOPT$. 
This completes our proof and demonstrates that the minimum among them is at most $\alf \cdot \cOPT$ since any weighted average of a set of values is greater than or equal to their minimum.

Denoting $\wx$, $\wy$, and $\wz$ as the weights of solutions $\GW$, $\ST$, and $\IT$ respectively, let $\cWAG$ represent their weighted average cost. 
As we are taking an average, we assume $\wx + \wy + \wz = 1$ to simplify the calculation. 
We also have $\wx, \wy, \wz \ge 0$.
The bound for the weighted average is then given by
\begin{align*}
    \cWAG \le{}& (\alf \cdot \wx + \alf \cdot \wy + \alf \cdot \wz) \cdot \cOPT \\
    & + ((2 - \alf) \cdot \wx + (\p - \alf) \cdot \wy) \cdot \va \\
    & + ((2 - \alf) \cdot \wx + (\p + \bta - \alf) \cdot \wy + (\bta - \alf) \cdot \wz) \cdot \vba \\
    & + ((2 - 2\alf) \cdot \wx + (2\p + \bta - 2\alf) \cdot \wy + \bta \cdot \wz) \cdot \vbb \\
    & + ((2 - \alf \bta) \cdot \wx + (2\p - \alf \bta) \cdot \wy) \cdot \vc \\
    & + ((2 - \alf \bta) \cdot \wx + (2\p + \bta - \alf \bta) \cdot \wy + (\bta - \alf \bta) \cdot \wz) \cdot \vd
\end{align*}

Given that $\wx + \wy + \wz = 1$, we have $(\alf \cdot \wx + \alf \cdot \wy + \alf \cdot \wz) \cdot \cOPT = \alf \cdot \cOPT$. 
Thus, the first term in the expression is $\alf \cdot \cOPT$.

To ensure $\cWAG \le \alf \cdot \cOPT$, we aim to make the rest of the expression non-positive. 
Since $\va$, $\vba$, $\vbb$, $\vc$, and $\vd$ are non-negative values, it suffices to make their coefficients non-positive by assigning suitable values to $\alf$, $\bta$, and the weights $\wx$, $\wy$, and $\wz$.
This leads to finding values that satisfy the following inequalities, with each inequality corresponding to one of the coefficients.
\begin{align*}
    &(2 - \alf) \cdot \wx + (\p - \alf) \cdot \wy \le 0 \tag{$\va$}\\
    &(2 - \alf) \cdot \wx + (\p + \bta - \alf) \cdot \wy + (\bta - \alf) \cdot \wz \le 0 \tag{$\vba$}\\
    &(2 - 2\alf) \cdot \wx + (2\p + \bta - 2\alf) \cdot \wy + \bta \cdot \wz \le 0 \tag{$\vbb$}\\
    &(2 - \alf \bta) \cdot \wx + (2\p - \alf \bta) \cdot \wy \le 0 \tag{$\vc$}\\
    &(2 - \alf \bta) \cdot \wx + (2\p + \bta - \alf \bta) \cdot \wy + (\bta - \alf \bta) \cdot \wz \le 0 \tag{$\vd$}
\end{align*}

We can also use a weighted average to ensure that our solution in the induction base has cost $\leq\alf\cdot\cOPT$. In this case, the $\IT$ solution cannot be employed as it represents the final step of recursion. So, we must have $\wz=0$. 
Additionally, it's essential to note that in this step, $\pi(\DGW) = \pi(\B \cup \D) = 0$, resulting in $\vba = \vbb = \vd = 0$. 
Thus, only the inequalities for the coefficients of $\va$ and $\vc$ remain relevant, which already do not contain $\wz$:
\begin{align*}
    &(2 - \alf) \cdot \wx + (\p - \alf) \cdot \wy \le 0 \tag{$\va$}\\
    &(2 - \alf \bta) \cdot \wx + (2\p - \alf \bta) \cdot \wy \le 0 \tag{$\vc$}
\end{align*}

We can see that if a solution for the system of constraints used for the induction step is found, setting
$\wz=0$ and scaling $\wx$ and $\wy$ by a factor of $\frac{1}{1-\wz}$ gives us a solution for these two new constraints with $\wx+\wy=1$ and $\wz=0$. So, whatever values of $\alf$ and $\bta$ we find by solving the initial system of inequalities will give us a valid solution and an approximation guarantee of $\alf$.

Considering the best-known approximation factor for the Steiner tree problem, which is $\p=\ln(4)+\epsilon$~\cite{DBLP:conf/stoc/ByrkaGRS10}, we determine that choosing the values $\alf=\apx$, $\bta=1.252$, $\wx=0.385$, $\wy=0.187$, and $\wz=0.428$ satisfies all the inequalities for a small enough value of $\epsilon$. 
This provides a valid proof for both the induction base and induction step, leading to the conclusion of the following theorem.
\begin{theorem}
    The minimum cost among $\GW$, $\ST$, and $\IT$ is a $\apx$-approximate solution for the Prize-Collecting Steiner Tree instance $\I$. Therefore, \rpcst~ is an $\apx$ approximation for PCST.
\end{theorem}

Finally, we note that our algorithm runs in polynomial time.

\begin{theorem}
    The procedure $\rpcst(\I)$ runs in polynomial time.
\end{theorem}
\begin{proof}
    The procedure $\rpcst(\I)$ calls the $\pcstgw(\Ib)$ which runs in polynomial time and a polynomial time algorithm \steinertree~ for the Steiner tree problem.
    Then it recursively calls itself on a new instance such that the new instance has more vertices with a penalty of $0$. The construction of this instance involves a simple loop on the vertices and is done in polynomial time.
    Since the number of vertices is $|\V|$, and each time the number of vertices with non-zero penalty decreases by one, the recursion depth is at most $|\V|$.
    So, in $\rpcst$ we have a polynomial number of recursive steps, and each step takes a polynomial amount of time. Therefore, the total running time of the algorithm is polynomial in the size of the input.
\end{proof}

\section{Necessities in Our Algorithm}
\label{sec:intuition}
In this section, we demonstrate the necessity of utilizing all three solutions in $\rpcst$ and selecting the minimum among them. 
Table~\ref{tabel:sign_coefficient} is completed based on the constraints $1 < p < \alf < 1.8$, derived from the NP-hardness of finding an exact algorithm for Steiner tree, the fact that Steiner tree is a special case of PCST, and the goal of achieving an approximation factor better than $1.8$.
Additionally, we select $\bta$ such that $2/\alf \le \bta \le \alf$ because if $2/\alf > \bta$, both coefficients in $\vc$ will be positive. 
Also, if $\bta > \alf$, all coefficients of $\vba$ become positive.

\begin{table}[H]
\begin{center}
\begin{tabular}{|c|c|c|c|c|c|}
\cline{2-6}
\multicolumn{1}{c|}{} & $\va$ & $\vba$ & $\vbb$ & $\vc$ & $\vd$ \\
\hline
\GW & + & + & - & - & -\\
\hline
\ST & - & + & +\footnotemark & ? & +\\
\hline
\IT & 0 & - & + & 0 & -\\
\hline
\end{tabular}

\caption{Sign of coefficients of each solution.}
\label{tabel:sign_coefficient}
\end{center}
\end{table}
\footnotetext{Could potentially turn negative after further improvement in the approximation factor of the Steiner tree problem.}

Table~\ref{tabel:sign_coefficient} demonstrates the sign of the coefficient for each variable in every algorithm. 
We refer to this table to explain why all three algorithms are essential.
We need to find a combination of these algorithms such that the weighted average of these coefficients adds up to zero.
Since each row associated with an algorithm has at least one positive value, achieving this balance is not possible if we use only one of the algorithms.
Moreover, omitting $\IT$ results in a positive coefficient for $\vba$, making the iterative approach necessary.
Similarly, using $\GW$ and $\IT$ together leads to a positive coefficient for $\va$, emphasizing the need for $\ST$ to offset it.
Lastly, if we drop $\GW$, the coefficient of $\vbb$ constrains our approximation factor, as its coefficient in the $\IT$ algorithm is positive, and in the $\ST$ algorithm, it is $2\p + \bta - 2 \alf$.
Given that the best-known approximation factor for the Steiner tree is $\ln(4)+\epsilon$~\cite{DBLP:conf/stoc/ByrkaGRS10}, replacing $\p$ with $\ln(4)+\epsilon$ results in a positive value for the coefficient of $\vbb$ in the $\ST$ algorithm.
Therefore, the $\GW$ algorithm is necessary to decrease the coefficient of $\vbb$.

\subsection{Bad example for \boldmath$\bta > 2$}
\label{sec:bta2}
Let $\bta=2(1+\epsilon)$ for some $\epsilon>0$. We consider a star graph $\G$ with $n+1$ vertices as shown in Figure \ref{fig:bta2}, where one vertex is a central vertex and all other vertices are connected to this vertex with edges of length $1$ for some value of $n$ such that $\frac{1}{n-1}<\epsilon$. We construct an instance of PCST on this graph where one of the non-central vertices is the root, the central vertex has penalty $0$, and any other vertex has penalty $2(1+\frac{1}{n-1})$. 

When we run the \GW~algorithm on this instance, the center vertex dies instantly as it has $0$ coloring potential. Additionally, as $\frac{1}{n-1}<\epsilon$, each non-root leaf has coloring potential 
$$\frac{2(1+\frac{1}{n-1})}{\bta}=\frac{(1+\frac{1}{n-1})}{1+\epsilon}<1$$ and therefore dies before reaching the central vertex. So, the $\GW$ solution will pay penalty $(n-1)(2(1+\frac{1}{n-1}))=2n$. This is twice the cost of the optimal solution, which can be obtained by taking all $n$ edges of length $1$. The other solutions we consider will also have the same cost as the $\GW$ solution, as they will aim to connect only the $\Root$ and will pay the penalties for all the dead vertices. So, using any $\bta>2$ will lead to an approximation factor of at least $2$.

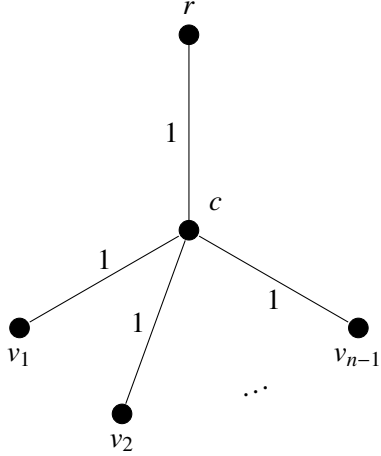
\begin{figure}
\begin{center}
\begin{tikzpicture}[scale=1.3]
\def\r{0.1}
\foreach \v/\deg/\d/\l/\dir in {r/90/2/r/90, c/0/0/c/45, v1/210/2/v_1/270, v2/250/2/v_2/270, v3/330/2/v_{n-1}/270} {
    \node[label={\dir:$\l$}] (\v) at (\deg:\d) {};
    \draw[fill=black] (\v) circle (\r);
};
\node[label={.}] at (287:2) {};
\node[label={.}] at (290:2) {};
\node[label={.}] at (293:2) {};
\draw (r) -- node[left]{$1$} (c);
\draw (c) -- node[above]{$1$} (v1);
\draw (c) -- node[left]{$1$} (v2);
\draw (c) -- node[below]{$1$} (v3);
\end{tikzpicture}
\end{center}
\caption{A star graph with $n+1$ vertices. We construct a PCST instance on this graph with vertex $r$ as the root, the central vertex $c$ having penalty $0$, and all other vertices with having penalty $2(1+\frac{1}{n-1})$.}
\label{fig:bta2}
\end{figure}

\section{Acknowledgements}

Partially supported by DARPA QuICC, ONR MURI 2024 award on Algorithms, Learning, and Game Theory, Army-Research Laboratory (ARL) grant W911NF2410052, NSF AF:Small grants 2218678, 2114269, 2347322.


\bibliographystyle{abbrv}
\bibliography{references}

\newpage 
\appendix

\section{Proofs of GW Algorithm}

\newtheorem*{recall:lemma:gwtree}{Lemma~\ref{lemma:gwtree}}
\begin{recall:lemma:gwtree}
    Let $\T$ be the tree returned by Algorithm \ref{alg:gw}. 
    We can bound the total weight of this tree by
    \begin{align*}
        \cc(\T)\leq 2\cdot\sum_{\substack{\SET \subset \V-\RSET;\\\SET\cap\V(\T)\neq\emptyset}}\ys=2\cdot\sum_{v \in \V(\T)-\RSET}\yv\text{.}
    \end{align*}
\end{recall:lemma:gwtree}
\begin{proof}
    First, we note that the $T$ is completely colored in our algorithm. Therefore, we have
    \begin{align*}
        \cc(\T)&=\sum_{\substack{\SET \subset \V}} \ys \cdot\left|\delta(\SET)\cap\T\right|\\
        &=\sum_{\substack{\SET \subset \V;\\\SET\cap\V(\T)\neq\emptyset}} \ys \cdot\left|\delta(\SET)\cap\T\right| \tag{$\left|\delta(\SET)\cap\T\right|=0$ if $\SET\cap\V(\T)=\emptyset$} \text{.}
    \end{align*}
    So we first want to prove the following:
    \begin{align*}
        \sum_{\substack{\SET \subset \V;\\\SET\cap\V(\T)\neq\emptyset}} \ys \cdot\left|\delta(\SET)\cap\T\right|
        \leq 2\cdot\sum_{\substack{\SET \subset \V-\RSET;\\\SET\cap\V(\T)\neq\emptyset}}\ys \text{.}
    \end{align*}
    We consider how each side of this inequality changes in one step of the coloring process. 
    Let $\CC$ be the set of components in this step. Take $H$ to be the tree obtained from $\T$ by contracting each component in $\CC$ to one vertex. For a set $\SET\in\CC$, let $\comp(\SET)$ denote the vertex corresponding to $\SET$ in $H$. We can see that 
    $\left|\delta(\SET)\cap\T\right|={deg}_H(\comp(\SET))$. 
    
    Let $\ah$ and $\ih$ respectively denote the set of active and inactive sets in $\CC$ that share a node of $\V(\T)$. Then each node in $H$ corresponds to a set in $\ah\cup\ih$. Now, assume that in this step, we are increasing the $\ys$ values for all active sets by $\Delta$. Then the increase in the left-hand side of the inequality is 
    $\Delta\cdot\sum_{\SET\in\ah}deg_H(\comp(\SET))$. 
    The increase in the right-hand side is equal to $2\Delta\cdot(\left|\ah\right|-1)$ as there are $\left|\ah\right|$ active sets intersecting $\V(\T)$ that increase their $\ys$ value by $\Delta$ and only one contains $\Root$. So it suffices to prove that $\sum_{\SET\in\ah}deg_H(\comp(\SET)) \leq 2(\left|\ah\right|-1)$.

    We argue that any leaf in $H$ must correspond to an active set in $\ah$. Assume otherwise that a leaf in $H$ is inactive in this step. So this leaf corresponds to a dead set $\SET$. As the corresponding vertex is a leaf, $\left|\delta(\SET)\cap\T\right|$ is equal to $1$. But in this case, $\SET$ would have been disconnected from $\Root$ in the second phase and this edge along with all edges inside $\SET$ would have been removed. So any leaf in $H$ must correspond to an active set and therefore vertices in $H$ not corresponding to active sets have a degree of at least $2$. Therefore, we have
    \begin{align*}
        \sum_{\SET\in\ah}deg_H(\comp(\SET)) &=  \sum_{v\in\V(H)}deg_H(v) - 
        \sum_{\SET\in\ih}deg_H(\comp(\SET))\\
        &\leq 2(\left|\V(H)\right|-1) - 
        \sum_{\SET\in\ih}deg_H(\comp(\SET)) \tag{$H$ is a tree}\\
        &\leq 2(\left|\V(H)\right|-1) - 
        2(\left|\ih\right|) \tag{vertices of $\ih$ have degree at least $2$}\\
        &\leq 2(\left|\V(H)\right|-1) - 
        2(\left|\V(H)\right|-\left|\ah\right|) \tag{$\left|\ih\right|=\left|\V(H)\right|-\left|\ah\right|$}\\
        &\leq 2(\left|\ah\right|-1) \text{.}
    \end{align*}
    This completes the first part of our proof.
    Next, we need to show that 
    \begin{align*}    
    \sum\limits_{\substack{\SET \subset \V;\Root\not\in\SET;\\\SET\cap\V(\T)\neq\emptyset}}\ys=\sum\limits_{v \in \V(\T);v\neq \Root}\yv.
    \end{align*}
    We expand the left-hand side:
    \begin{align*}
        \sum\limits_{\substack{\SET \subset \V-\RSET\\\SET\cap\V(\T)\neq\emptyset}}\ys=\sum\limits_{\substack{\SET \subset \V-\RSET\\\SET\cap\V(\T)\neq\emptyset}}\sum_{v\in\SET} \ysv\text{.}
    \end{align*}
    Now, for a set $\SET$ where $\SET\cap\V(\T)\neq\emptyset$, we claim that $\ysv$ can only be non-zero if $v\in\V(\T)$. Assume otherwise that $\ysv>0$ but $v\not\in\V(\T)$. This means that $v$ must have been disconnected from $\Root$ as part of a dead set. As $\ysv>0$, $v$ was not already dead when becoming part of $\SET$. But any active set containing $v$ afterward also includes all the vertices in $\SET$. So any dead set that could lead to $v$ being disconnected from $\Root$ also disconnects all the vertices in $\SET$. This goes against our assumption that $\SET\cap\V(\T)\neq\emptyset$, so if $\ysv>0$, $v$ must be in $\V(\T)$. Based on this, we can rewrite the previous equation as follows:
    \begin{align*}
        \sum\limits_{\substack{\SET \subset \V-\RSET;\\\SET\cap\V(\T)\neq\emptyset}}\ys&=\sum\limits_{\substack{\SET \subset \V-\RSET\\\SET\cap\V(\T)\neq\emptyset}}\sum_{v\in\SET} \ysv\\
        &=\sum\limits_{\substack{\SET \subset \V-\RSET\\\SET\cap\V(\T)\neq\emptyset}}\sum_{v\in\V(\T)\cap\SET} \ysv \tag{$\ysv=0$ if $v\not\in\V(\T)$}\\
        &=\sum_{v\in\V(\T)-\RSET}\sum\limits_{\substack{\SET \subset \V-\RSET\\v\in\SET}} \ysv \tag{Change the order of summation noting $\Root\not\in\SET$}\\
        &=\sum_{v\in\V(\T)-\RSET} \yv. \tag{Definition \ref{def:coldur}}
    \end{align*}
    In the last equation, we use the fact that if a set $\SET$ contains $\Root$, all the coloring for that set will be assigned to $\Root$ so $\ysv=0$ for all $v\neq\Root$. This completes our proof of the lemma.
\end{proof}

\newtheorem*{recall:lemma:gwdead}{Lemma~\ref{lemma:gwdead}}
\begin{recall:lemma:gwdead}
    For any vertex $v \in \V$, we have $\yv \leq \pi(v)$. 
    Furthermore, if $v \in \dead$ which means is a dead vertex, we have $\yv=\pi(v)$.
\end{recall:lemma:gwdead}
\begin{proof}
    By definition, $\yv$ is the amount of coloring potential of a vertex that is used and therefore $\yv\leq\pi(v)$ for all vertices. In addition, for a dead vertex $v$, its coloring potential is completely used so $\yv=\pi(v)$. 
\end{proof}

\newtheorem*{recall:lemma:gwlive}{Lemma~\ref{lemma:gwlive}}
\begin{recall:lemma:gwlive}
    Any vertex $v \not\in \V(\T)$ is a dead vertex.
\end{recall:lemma:gwlive}
\begin{proof}
    This is true by our construction. When we remove the single edge connecting a dead set to our tree, this only disconnects this dead set from the root. As no dead set contains live vertices, they all remain in $\V(\T)$. So any vertex outside $\V(\T)$ will be a dead vertex.
\end{proof}

\end{document}